\documentclass[10pt,twocolumn,twoside]{IEEEtran}

\usepackage{amsmath}
\usepackage{amssymb}
\usepackage{amsthm}
\usepackage{graphicx}
\usepackage{epstopdf}
\usepackage{bm}
\usepackage{cite}
\usepackage{subfigure}
\usepackage{color}

\newtheorem{theorem}{Theorem}[section]
\newtheorem{lemma}[theorem]{Lemma}
\newtheorem{corollary}[theorem]{Corollary}

\newtheorem{remark}[theorem]{Remark}
\newtheorem{example}[theorem]{Example}

\allowdisplaybreaks

\begin{document}
\title{Sensor Scheduling in Variance Based Event Triggered Estimation with Packet Drops} 

\author{Alex S. Leong, Subhrakanti Dey, and Daniel E. Quevedo
 \thanks{A preliminary version of parts of this work was presented at the 2015 European Control Conference, Linz, Austria \cite{LeongDeyQuevedo_ECC}.}
 \thanks{A. Leong and D. Quevedo are with the Department of Electrical Engineering (EIM-E), Paderborn University, 33098 Paderborn, Germany. 
   E-mail: {\tt alex.leong@upb.de, dquevedo@ieee.org}. S. Dey is with the Department of Engineering Science, Uppsala University, Uppsala, Sweden. E-mail: {\tt Subhra.Dey@signal.uu.se.} }
\thanks{This work was supported by the Australian Research Council under grants  DE120102012 and DP120101122.}  
 }

\maketitle

\begin{abstract}
This paper considers a remote state estimation problem with multiple sensors observing a dynamical 
process, where sensors transmit local state estimates over an independent and identically distributed (i.i.d.) packet dropping channel to a remote estimator. At every discrete time instant, the remote estimator decides whether each sensor should transmit or not, with each sensor transmission incurring a fixed energy cost. The channel is shared such that collisions will occur if more than one sensor transmits at a time. Performance is quantified via an optimization problem that minimizes a convex combination of the expected estimation error covariance at the remote estimator and expected energy usage across the sensors. For transmission schedules dependent only on the estimation error covariance at the remote estimator, this work establishes structural results on the optimal scheduling which show that 1) for unstable systems, if the error covariance is large then a sensor will always be scheduled to transmit, and 2) there is a threshold-type behaviour in switching from one sensor transmitting to another. Specializing to the single sensor case, these structural results demonstrate that a threshold policy (i.e. transmit if the error covariance exceeds a certain threshold and don't transmit otherwise) is optimal. We also consider the situation where sensors transmit measurements instead of state estimates, and establish structural results including the optimality of threshold policies for the single sensor, scalar case. These results provide a theoretical justification for the use of such threshold policies in variance based event triggered estimation. Numerical studies confirm the qualitative behaviour predicted by our structural results.  An extension of the structural results to Markovian packet drops is also outlined.
\end{abstract}

\section{Introduction}
The concept of event triggered estimation of dynamical systems, where sensor measurements or state estimates are sent to a remote estimator/controller only when certain events occur, has gained significant recent attention. By transmitting only when necessary, as dictated by performance objectives,  e.g., such as when the estimation quality at the remote estimator has deteriorated sufficiently, potential savings in energy usage can be achieved, which are important in networked estimation and control applications. 

\emph{Related Work}: Event triggered estimation has been investigated in e.g. \cite{XuHespanha_comm_logic,ImerBasar,CogillLallHespanha,LiLemmonWang,TrimpeDAndrea_IFAC,WeimerAraujoJohansson,SijsLazar,TrimpeDAndrea_journal,XiaGuptaAntsaklis,WuJiaJohanssonShi,Han_event_trigger,Trimpe_CDC},  while event triggered control has also been studied in  e.g. \cite{AstromBernhardsson,Tabuada,RabiJohansson,RameshSandbergJohansson,Quevedo_event_trigger}. 
Many rules for deciding when a sensor should transmit have been proposed in the literature, such as if the estimation error \cite{ImerBasar,LiLemmonWang,WeimerAraujoJohansson,XiaGuptaAntsaklis}, error in predicted output \cite{TrimpeDAndrea_IFAC,Trimpe_CDC}, other functions of the estimation error \cite{CogillLallHespanha,WuJiaJohanssonShi,Han_event_trigger},  or the error covariance \cite{TrimpeDAndrea_journal}, exceeds a given threshold. These transmission policies often lead to energy savings. However, the motivation for using these rules are usually based on heuristics. 
Another gap in  current literature on event triggered estimation is that mostly the idealized case, where all  transmissions (when scheduled) are  received at the remote estimator, is considered. Packet drops \cite{Sinopoli}, which are unavoidable when using a wireless communication medium, are neglected in these works, save for some works in event triggered control \cite{RabiJohansson,Quevedo_event_trigger}.

In a different line of research, sensor scheduling problems, where one wants to determine a schedule such that at each time instant, one or more sensors are chosen to transmit in order to minimize an expected error covariance performance measure, have been extensively studied, see  e.g. \cite{GuptaChungHassibiMurray,ShiChengChen,MoGarone,ShiZhang,Huber}. However, these schedules are often constructed ahead of time in an offline manner and do not take into account random packet drops or  variations in the state estimates, i.e. are not event triggered. Covariance based switching for scheduling between two sensors was investigated in \cite{SandbergRabiSkoglundJohansson}.  Structural results were derived for infinite horizon sensor scheduling problems in \cite{MoGaroneSinopoli,ZhaoZhangHu}, which showed that optimal schedules are independent of initial conditions and can be approximated arbitrarily closely with  periodic schedules of finite length, with \cite{MoGaroneSinopoli} also extending these results to networks with packet drops.

\emph{Summary of Contributions}: 
In this paper, we consider a multi-sensor event triggered estimation problem with i.i.d. packet drops, and derive structural properties on the optimal transmission schedule. In particular, the main contributions of this paper are:
\begin{itemize}
\item In contrast to previous works on event-triggered estimation, we allow for the more practical situation where sensor transmissions  experience random packet drops.  
\item Rather than specifying the form of the transmission schedule a priori, in this work the transmission decisions are determined by solving an optimization problem that minimizes a convex combination of the expected error covariance and expected energy usage.
\item We derive \emph{structural results} on the form of the subsequent optimal transmission schedule. For transmission schedules which decide whether to transmit local state estimates based only on knowledge of the error covariance at the remote estimator, our analysis  shows that 1) for unstable systems, if the error covariance is large, then a sensor will always be scheduled to transmit, and 2) there is a threshold-type behaviour in switching from one sensor transmitting to another. 
\item Specializing these structural results to the single sensor case shows that a \emph{threshold policy}, where the sensor transmits if the error covariance exceeds a threshold and does not transmit otherwise, is optimal. This result has also been proved different techniques in our conference contribution \cite{LeongDeyQuevedo_ECC}, and, in a related setup, in \cite{MoSinopoliShiGarone}. For noiseless measurements and no packet drops, similar structural results were derived using majorization theory for scalar \cite{LipsaMartins} and vector \cite{NayyarBasar} systems respectively. 
\item In the situation where sensor measurements (rather than local estimates) are transmitted, related structural results are derived, in particular  the optimality of threshold policies in the single sensor, scalar case. These structural results provide a theoretical justification for the use of such variance based threshold policies in event triggered estimation.
 However, for vector systems, we provide counterexamples to show that in general threshold-type policies are not optimal.  
\item The structural results are extended to Markovian packet drops, where we show that for a single sensor there exist in general two different thresholds, depending on whether packets were dropped or received at the previous time instant.
\end{itemize}

The remainder of this paper is organized as follows. 
 Section \ref{model_sec} presents the system model, while the optimization problems are formulated in Section \ref{optimization_prob_sec}. Structural results on the optimal transmission scheduling are derived in Sections \ref{finite_horizon_structural_sec} and \ref{infinite_horizon_structural_sec}. The special case of a single sensor is then studied in Section \ref{single_sensor_sec}. The situation where sensor measurements are transmitted is studied in Section \ref{tx_meas_sec}. 
Numerical studies, including comparisons of our approach with schemes where transmission decisions are made using current sensor measurements, are presented in Section \ref{numerical_sec}. 
An extension of our structural results to Markovian packet drops is outlined in Section \ref{markovian_sec}. 

\section{System model and Remote Estimation Schemes}
\label{model_sec}
A diagram of the system model is shown in Fig. \ref{system_model}.
Consider a discrete time process 
\begin{equation}
\label{state_eqn}
x_{k+1}=A x_k + w_k
\end{equation}
where $x_k \in \mathbb{R}^n$ and $w_k$ is i.i.d. Gaussian with zero mean and covariance $Q$. 
There are $M$ sensors, with each sensor having measurements
\begin{equation}
\label{measurement_eqn}
y_{m,k} = C_m x_k + v_{m,k}, \quad m \in \{1,\dots,M\}
\end{equation}
where $y_{m,k} \in \mathbb{R}^{n_m}$ and $v_{m,k}$ is Gaussian with zero mean and covariance $R_m$. 
\begin{figure}[tb!]
\centering 
\includegraphics[scale=0.4]{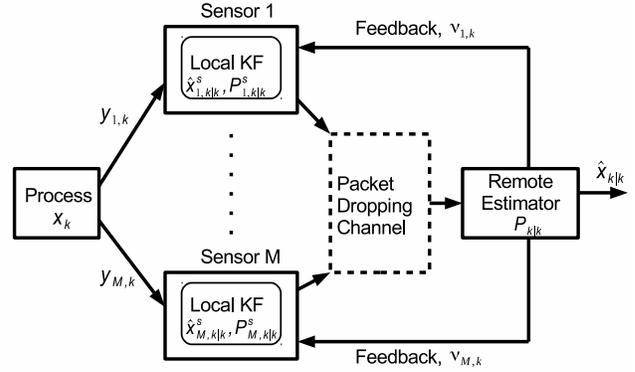} 
\caption{System model}
\label{system_model}
\end{figure}
The noise processes $\{w_k\}, \{v_{m,k}\}, m=1,\dots,M$ are assumed to be mutually independent. 

Each sensor has some computational capability and can run a local Kalman filter. 
The local state estimates and error covariances
\begin{equation*}
\begin{split}
\hat{x}_{m,k|k-1}^s & \triangleq \mathbb{E}[x_k|y_{m,0},\dots,y_{m,k-1}] \\
\hat{x}_{m,k|k}^s & \triangleq \mathbb{E}[x_k|y_{m,0},\dots,y_{m,k}] \\
P_{m,k|k-1}^s & \triangleq  \mathbb{E}[(x_k-\hat{x}_{m,k|k-1}^s)(x_k-\hat{x}_{m,k|k-1}^s)^T\\ & \quad\quad\quad|y_{m,0},\dots,y_{m,k-1}]\\
P_{m,k|k}^s & \triangleq  \mathbb{E}[(x_k-\hat{x}_{m,k|k}^s)(x_k-\hat{x}_{m,k|k}^s)^T|y_{m,0},\dots,y_{m,k}]
\end{split}
\end{equation*}
can be computed using the standard  Kalman filtering equations at sensors $m=1,\dots,M$. For the results in Sections \ref{model_sec}-\ref{structural_properties_sec}, we will assume that each pair $(A,C_m)$ is detectable and the pair $(A,Q^{1/2})$ is stabilizable. In Section \ref{tx_meas_sec} we will relax this assumption when we consider transmission of sensor measurements, and consequently only detectability of the overall system is required. 
Let $\bar{P}_m^s$ be the steady state value of $P_{m,k|k-1}^s$, and $\bar{P}_m$ be the steady state value of $P_{m,k|k}^s$, as $k \rightarrow \infty$, which both exist due to the detectability assumptions. 

 Let $\nu_{m,k} \in \{0,1\}, m=1,\dots,M$ be  decision variables such that $\nu_{m,k}=1$ if and only if $\hat{x}_{m,k|k}^s$ is to be transmitted to the remote estimator at time $k$. 
 Transmitting state estimates when there are packet drops generally gives better estimation performance than transmitting measurements \cite{XuHespanha,Schenato}, and in the case of a single sensor is the best non-causal strategy \cite{GuptaHassibiMurray}. 
 We will focus on the situation where $\nu_{m,k}$ are computed at the remote estimator at time $k-1$ and communicated to the sensors without error via feedback links before transmission at the next time instant $k$,\footnote{This requires synchronization between each sensor and the remote estimator, though not between individual sensors. Note that in wireless communications,
online computation of powers at the base station which is then fed back to the mobile transmitters is commonly done in practice \cite{Molisch}, at time scales on the order of milli-seconds.} see Section \ref{imperfect_feedback_sec} on how to take into account losses in the feedback links. 
Since our interest lies in decision making at the remote estimator, we shall assume that the decisions $\nu_{m,k}$ do not depend on the current value of $x_k$ (or functions of $x_k$ such as measurements and local state estimates). In particular, in this paper we will assume that $\nu_{m,k}$ depends only on the error covariance at the remote estimator, similar to the variance based triggering schemes of \cite{TrimpeDAndrea_journal}, see Section \ref{optimization_prob_sec}. 
 
At time instances when $\nu_{m,k}=1$, sensor $ m$ transmits its local state estimate $\hat{x}_{m,k|k}^s$ over a packet dropping channel.
Let $\gamma_{m,k}, m=1,\dots,M$ be random variables such that $\gamma_{m,k}=1$ if the transmission from sensor $m$ at time $k$ is successfully received by the remote estimator, and $\gamma_{m,k}=0$ otherwise. It is assumed that the channel is shared such that if more than one sensor transmits at any time, then collisions will occur. Thus, $\gamma_{m,k}=0$ and $\gamma_{n,k} = 0$ with probability one if both $\nu_{m,k} = \nu_{n,k}=1$.
We will assume that $\{\gamma_{m,k}\}$ are i.i.d. Bernoulli with 
$$\mathbb{P}(\gamma_{m,k}=1) = \lambda_m, \quad m=1,\dots,M.$$
See Section \ref{markovian_sec} for  some results with Markovian packet drops.

\subsection{Optimal Remote Estimator}
At instances where $\nu_{m,k}=1$, it is assumed that the remote estimator knows whether the transmission was successful or not, i.e., the remote estimator knows the value $\gamma_{m,k}$. While if  $\nu_{m,k}=0$, since sensor $m$ is not scheduled to transmit at this time, the corresponding $\gamma_{m,k}$ is assumed to be of no use to the remote estimator. We can define 
\begin{equation*}
\begin{split}
\mathcal{I}_k \triangleq & \{\nu_{1,0},\dots,\nu_{1,k}, \nu_{1,0} \gamma_{1,0},\dots,\nu_{1,k} \gamma_{1,k},\\ &\quad \nu_{1,0} \gamma_{1,0} \hat{x}_{1,0|0}^s,\dots,\nu_{1,k} \gamma_{1,k} \hat{x}_{1,k|k}^s, \dots\dots,
\\ & \quad \nu_{M,0},\dots,\nu_{M,k},\nu_{M,0} \gamma_{M,0},\dots,\nu_{M,k} \gamma_{M,k}, \\ &\quad \nu_{M,0} \gamma_{M,0} \hat{x}_{M,0|0}^s,\dots,\nu_{M,k} \gamma_{M,k} \hat{x}_{M,k|k}^s\}
\end{split}
\end{equation*}
as the information set available to the remote estimator at time $k$. 
 Denote the state estimates and error covariances at the remote estimator by:
 \begin{equation*}
\begin{split}
\hat{x}_{k|k} & \triangleq \mathbb{E}[x_k| \mathcal{I}_k] \\
\hat{x}_{k+1|k} & \triangleq \mathbb{E}[x_{k+1}| \mathcal{I}_k] \\
P_{k|k} & \triangleq  \mathbb{E}[(x_k-\hat{x}_{k|k})(x_k-\hat{x}_{k|k})^T|\mathcal{I}_k] \\
P_{k+1|k} & \triangleq  \mathbb{E}[(x_{k+1}-\hat{x}_{k+1|k})(x_{k+1}-\hat{x}_{k+1|k})^T|\mathcal{I}_k].
\end{split}
\end{equation*}

If a sensor $\breve{m} \in \{1,\dots,M\}$ has been scheduled by the remote estimator to transmit at time $k$,\footnote{Since collisions occur if more than one sensor transmits at the same time, we clearly should not schedule more than one sensor to transmit at a time.} 
then the state estimates and error covariances at the remote estimator
are updated as follows:
\begin{align}
\label{multi_sensor_optimal_estimator}
\hat{x}_{k+1|k} & = A \hat{x}_{k|k}  \nonumber \\
\hat{x}_{k|k} & = \hat{x}_{k|k-1} + \gamma_{\breve{m},k} K_{\breve{m},k} (\hat{x}_{\breve{m},k|k}^s - \hat{x}_{k|k-1}) \nonumber \\
P_{k+1|k} & = A P_{k|k} A^T + Q  \nonumber \\
P_{k|k} & = (I - \gamma_{\breve{m},k} K_{\breve{m},k}) P_{k|k-1} (I - \gamma_{\breve{m},k} K_{\breve{m},k})^T  \nonumber\\ & \!\!\!\!\!\!\!\!  + \gamma_{\breve{m},k}  (I\!-\!\gamma_{\breve{m},k} K_{\breve{m},k}) P_{0\breve{m},k} (I \!-\! K^s_{\breve{m},k} C_{\breve{m}})^T K_{\breve{m},k}^T   \nonumber \\
& \!\!\!\!\!\!\!\! + \gamma_{\breve{m},k}  K_{\breve{m},k} (I \!-\! K^s_{\breve{m},k} C_{\breve{m}}) P_{0\breve{m},k}^T (I\!-\!\gamma_{\breve{m},k} K_{\breve{m},k})^T  \nonumber \\ & \!\!\!\!\!\!\!\! + \gamma_{\breve{m},k}  K_{\breve{m},k} (I\!-\!K^s_{\breve{m},k} C_{\breve{m}}) P_{\breve{m},k|k-1}^s (I\!-\!K^s_{\breve{m},k} C_{\breve{m}})^T K_{\breve{m},k}^T  \nonumber \\
& \!\!\!\!\!\!\!\! + \gamma_{\breve{m},k}  K_{\breve{m},k} K^s_{\breve{m},k} R_{\breve{m}} K^{s T}_{\breve{m},k} K_{\breve{m},k}^T   \nonumber \\
P_{0\breve{m},k+1} & = A (I\!-\!\gamma_{\breve{m},k} K_{\breve{m},k}) P_{0\breve{m},k} (I\!-\!K^s_{\breve{m},k} C_{\breve{m}})^T A^T  \nonumber \\ & \!\!\!\!\!\!\!\!+\! \gamma_{\breve{m},k} A K_{\breve{m},k} (I\!-\!K^s_{\breve{m},k} C_{\breve{m}}) P_{\breve{m},k|k-1}^s (I - K^s_{\breve{m},k} C_{\breve{m}})^T A^T  \nonumber \\ 
&\!\!\!\!\!\!\!\! + Q + \gamma_{\breve{m},k} A K_{\breve{m},k} K^s_{\breve{m},k} R_{\breve{m}} K^{s T}_{\breve{m},k} A^T  \nonumber \\
P_{0m,k+1} & = A (I\!-\!\gamma_{\breve{m},k} K_{\breve{m},k}) P_{0m,k} (I\!-\!K^s_{m,k} C_m)^T A^T  \nonumber \\ & \!\!\!\!\!\!\!\!+\! \gamma_{\breve{m},k} A K_{\breve{m},k} (I\!-\!K^s_{\breve{m},k} C_{\breve{m}}) P_{\breve{m}m,k} (I \!-\! K^s_{m,k} C_m)^T A^T 
 \nonumber \\ &  \!\!\!\!\!\!\!\! +Q, \quad m \neq \breve{m}  \nonumber \\
P_{mn,k+1} & = A(I-K^s_{m,k} C_m) P_{mn,k} (I-K^s_{n,k} C_{n})^T A^T + Q,  \nonumber \\ & \quad\quad m, n > 0, m \neq n
%
\end{align}
where $K^s_{m,k}\triangleq P_{m,k|k-1}^s C_m^T (C_m P_{m,k|k-1}^s C_m^T + R_m)^{-1}$ is the local Kalman filter gain of sensor $m$ at time $k$, $K_{\breve{m},k} \! = I$ if $P_{k|k-1} \!- \! P_{0\breve{m},k}(I\!-\!K^s_{\breve{m},k} C_{\breve{m}})^T \! - \! (I \!-\!K^s_{\breve{m},k} C_{\breve{m}}) P_{0\breve{m},k}^T  + (I-K^s_{\breve{m},k} C_{\breve{m}}) P_{\breve{m},k|k-1}^s (I-K^s_{\breve{m},k} C_{\breve{m}})^T  + K^s_{\breve{m},k} R_{\breve{m}} K^{sT}_{\breve{m},k} = P_{k|k-1} - P_{0\breve{m},k} (I - K^s_{\breve{m},k} C_{\breve{m}})^T$, and $K_{\breve{m},k} \! = \Big(P_{k|k-1} - P_{0\breve{m},k} (I - K^s_{\breve{m},k} C_{\breve{m}})^T \Big) \Big(P_{k|k-1} - P_{0\breve{m},k}(I-K^s_{\breve{m},k} C_{\breve{m}})^T - (I-K^s_{\breve{m},k} C_{\breve{m}}) P_{0\breve{m},k}^T  + (I-K^s_{\breve{m},k} C_{\breve{m}}) P_{\breve{m},k|k-1}^s (I-K^s_{\breve{m},k} C_{\breve{m}})^T + K^s_{\breve{m},k} R_{\breve{m}} K^{sT}_{\breve{m},k} \Big)^{-1}$ otherwise.
  The last three equations in (\ref{multi_sensor_optimal_estimator}) compute the quantities:
\begin{equation*}
\begin{split}
P_{0m,k} & \triangleq \mathbb{E}[(x_k-\hat{x}_{k|k-1})(x_k-\hat{x}_{m,k|k-1}^s)^T|\mathcal{I}_k]\\
P_{m0,k} & \triangleq \mathbb{E}[(x_k-\hat{x}_{m,k|k-1}^s)(x_k-\hat{x}_{k|k-1})^T|\mathcal{I}_k] \\
P_{mn,k} & \triangleq \mathbb{E}[(x_k-\hat{x}_{m,k|k-1}^s)(x_k-\hat{x}_{n,k|k-1}^s)^T|\mathcal{I}_k] \\
\end{split}
\end{equation*}
for $ m,n=1,\dots,M,$ where we note that $P_{0n,k} = P_{n0,k}^T$, and  $P_{nn,k} = P_{n,k|k-1}^s$. 

If no sensors are scheduled to transmit at time $k$, then the state estimates and error covariances are simply updated by:
\begin{equation}
\label{remote_estimator_no_tx} 
\begin{split}
\hat{x}_{k+1|k} & = A \hat{x}_{k|k}, \quad
\hat{x}_{k|k}  = \hat{x}_{k|k-1}  \\
P_{k+1|k} & = A P_{k|k} A^T + Q, \quad
P_{k|k}  =  P_{k|k-1}, \\
P_{0m,k+1} & = A  P_{0m,k} (I\!-\!K^s_{m,k} C_m)^T A^T + Q, \quad m=1,\dots,M
\end{split}
\end{equation}
The derivation of the optimal estimator equations (\ref{multi_sensor_optimal_estimator})-(\ref{remote_estimator_no_tx}) can be found in Appendix \ref{optimal_estimator_derivation}.

\begin{remark}
In (\ref{multi_sensor_optimal_estimator}), the terms $P_{0m,k+1}$ and $P_{mn,k+1}$ for $m, n \neq \breve{m}$ also need to be computed, since the scheduled sensor $\breve{m}$ will in general change over time. 
\end{remark}


\subsection{Suboptimal Remote Estimator}
\label{subopt_estimator_sec}
The estimator equations (\ref{multi_sensor_optimal_estimator}) are optimal, but difficult to analyze and derive structural results for. A suboptimal estimator that often performs well is a constant gain  estimator, which has the form (\ref{multi_sensor_optimal_estimator}) but with $K_{\breve{m},k} $ replaced by the constant gain $K_{\breve{m}} $ 
whenever sensor $\breve{m} \in \{1,\dots,M\}$ is scheduled to transmit. Suppose the constant gains $K_m, m=1,\dots,M$ are chosen using a similar procedure to \cite{Schenato},
 where $(P, P_{0m}, K_m)$ is a fixed point of the following set of equations:
\begin{align}
\label{fixed_point_equations}
P & = \lambda_m A (I - K_{m}) P (I -  K_{m})^T A^T   + (1-\lambda_m) A P A^T \nonumber\\ & \quad + \lambda_m  A (I- K_{m}) P_{0m} (I - K^s_m C_m)^T K_{m}^T  A^T \nonumber\\ 
& \quad + \lambda_{m}  A K_m (I - K^s_m C_m) P_{0m} (I- K_{m})^T A^T   \nonumber \\ 
& \quad + \lambda_{m}  A K_m (I - K^s_m C_m) \bar{P}_{m}^s (I- K^s_m C_m)^T K_m^T A^T   \nonumber \\ 
& \quad + \lambda_m A K_{m} K^s_m R_m K^{sT}_m K_{m}^T  A^T + Q\nonumber\\
P_{0m} & = A (I-\lambda_{m} K_{m}) P_{0m} (I-K^s_m C_m)^T A^T \nonumber\\ & \quad + \lambda_{m} A K_{m} (I-K^s_m C_m) \bar{P}_m^s (I - K^s_m C_m)^T A^T \nonumber\\ &\quad + Q + \lambda_{m} A K_{m} K^s_m R_m K^{sT}_m A^T \nonumber\\
K_{m} & = \Big(P - P_{0m} (I - K^s_m C_m)^T\Big)   \Big(P - P_{0m}(I-K^s_m C_m)^T \nonumber\\ & \quad - (I\!-\!K^s_m C_m) P_{0m}^T  \!+\! (I\!-\!K^s_m C_m) \bar{P}_m^s (I\!-\!K^s_m C_m)^T \nonumber\\ & \quad + K^s_m R_m K^{sT}_m \Big)^{-1},
\end{align}
with $K_m^s \triangleq \bar{P}_m^s C_m^T (C_m \bar{P}_m^s C_m^T + R_m)^{-1}$ being the steady state local Kalman gain of sensor $m$. The equations (\ref{fixed_point_equations}) are obtained by averaging over $\gamma_{\breve{m},k}$ in the recursion for $P_{k|k}$ (as well as the associated quantities $P_{0\breve{m},k}$ and $K_{\breve{m},k}$) in (\ref{multi_sensor_optimal_estimator}), and taking the steady state.

Then we have the following result:
\begin{theorem}
\label{fixed_point_theorem}
Suppose that $A$ is either (i) stable, or (ii) unstable but with  $\lambda_m > 1 - \frac{1}{\max_{i} |\sigma_i(A)|^2}, m = 1,\dots,M$, where $\sigma_i(A)$ is an eigenvalue of $A$. Then for each $m \in \{1,\dots,M\}$, amongst all possible constant gains $K_m$ satisfying $\max_{i} |\sigma_i(A(I-\lambda_m K_m))| < 1$,  there is a unique fixed point $(K_m, P_{0m}, P)$ to the set of equations (\ref{fixed_point_equations}) with 
$K_m=I, P_{0m} = \bar{P}_m^s$, and  $P$ being the unique solution to the equation
$P = (1-\lambda_m) A P A^T + Q + \lambda_m A (I-K^s_m C_m) \bar{P}_m^s (I-K^s_m C_m)^T A^T + \lambda_m A K^s_m R_m K^{sT}_m A^T.$
\end{theorem}

\begin{proof}
See Appendix \ref{fixed_point_theorem_proof}
\end{proof}
By Theorem \ref{fixed_point_theorem}, and in particular the fact that $K_m=I$ for each $m \in \{1,\dots,M\}$, the constant gain estimator $\tilde{x}_k$ with gains chosen by solving (\ref{fixed_point_equations}) is easily seen to simplify to the following: 
\begin{equation}
\label{remote_estimator_eqns_multi_sensor}
\begin{split}
\tilde{x}_{k|k} & = \left\{\begin{array}{ccc}  A \tilde{x}_{k-1|k-1} & , & \nu_{m,k} \gamma_{m,k} = 0 \\ \hat{x}_{m,k|k}^s & , & \nu_{m,k} \gamma_{m,k} = 1 \end{array}  \right. \\
\tilde{P}_{k|k} & = \left\{\begin{array}{ccl}  f(\tilde{P}_{k-1|k-1}) & ,  & \nu_{m,k} \gamma_{m,k} = 0\\ P_{m,k|k}^s & , &  \nu_{m,k} \gamma_{m,k} = 1  \end{array} \right. 
\end{split}
\end{equation} 
where 
\begin{equation}
\label{f_defn}
f(X) \triangleq A X A^T + Q.
\end{equation}
For the case of two sensors estimating independent Gauss-Markov systems, a similar estimator to (\ref{remote_estimator_eqns_multi_sensor}) was also studied in \cite{ShiZhang}. 
We now give some examples comparing the performance of the suboptimal estimator (\ref{remote_estimator_eqns_multi_sensor}) with the optimal estimator (\ref{multi_sensor_optimal_estimator}). 
Consider a two sensor system with parameters 
\begin{equation}
\label{estimator_comparison_params}
\begin{split}
&A = \left[\begin{array}{cc} 1.1 & 0.2 \\ 0.2 & 0.8  \end{array}\right],\, Q=I 
\end{split}
\end{equation} 
The other parameters are randomly generated: $C_1$ and $C_2$ are $1\times 2$ matrices with entries drawn from the uniform distribution $U(0.5,2)$, $R_1$ and $R_2$ are scalars drawn from $U(1,10)$, $\lambda_1$ and $\lambda_2$ are drawn from $U(0.5,1)$. 
The sensor that transmits is randomly chosen, with each sensor equally likely to be chosen.  Table \ref{estimator_comparison_table} gives  $\mathbb{E}[P_{k|k}]$ for the optimal (Opt.) and suboptimal (Subopt.) estimators for $20$ different randomly generated sets of parameters, where $\mathbb{E}[P_{k|k}]$ are obtained by taking the time average over a Monte Carlo simulation of length 100000. We also give values of $\mathbb{E}[P_{k|k}]$ for the case where measurements are transmitted (Tx. Meas.), which will be studied in Section \ref{tx_meas_sec}. We see that the suboptimal estimator often gives good performance when compared to the optimal estimator. 

\begin{table}
\caption{$\mathbb{E}[P_{k|k}]$ for different randomly generated sets of parameters}
\centering
\begin{tabular}{|c|c|c||c|c|c|} \hline
 Opt. &  Subopt.  &  Tx. Meas. &  Opt. &  Subopt.  &  Tx. Meas.  \\ \hline \hline
 3.1410 & 3.2216 & 3.2441   & 2.9906 & 3.0736 & 3.0705 \\ \hline
 3.9206 & 4.1434 & 4.2254   & 3.4654 & 3.6203 & 3.5358 \\ \hline
 3.6410 & 3.6990 & 3.8116   & 4.3822 & 4.6349 & 4.7211 \\ \hline
 3.2056 & 3.3040 & 3.3117   & 3.1704 & 3.2737 & 3.2766 \\ \hline
 4.9104 & 5.0146 & 5.1417   & 5.5227 & 5.6810 & 5.7757 \\ \hline
 3.5692 & 3.7251 & 3.7227   & 4.5079 & 4.6076 & 4.8558 \\ \hline 
 4.1598 & 4.2227 & 4.2522   & 3.9006 & 4.0154 & 4.0603 \\ \hline
 3.8327 & 3.9082 & 3.9827   & 3.2849 & 3.3567 & 3.3775 \\ \hline
 2.9210 & 3.0015 & 2.9704   & 7.0825 & 7.4793 & 7.8819 \\ \hline 
 3.7504 & 3.9277 & 3.9376   & 3.9697 & 4.1427 & 4.1661 \\ \hline 
\end{tabular}
\label{estimator_comparison_table}
\end{table}
Due to  its simplicity which makes it amenable to analysis, and its good performance in many cases,  we will concentrate on the estimator (\ref{remote_estimator_eqns_multi_sensor}) in Sections \ref{optimization_prob_sec}-\ref{structural_properties_sec}.

\begin{remark}
For the case  of a single sensor ($M=1$), the estimator (\ref{remote_estimator_eqns_multi_sensor}) corresponds to the optimal estimator, see, e.g., \cite{XuHespanha,Schenato}. 
\end{remark}

\subsection{Imperfect Feedback Links}
\label{imperfect_feedback_sec}
We have assumed that the feedback links are perfect, which models the most commonly encountered situation where the remote estimator has more resources than the sensors  and can transmit on the feedback links with very low probability of error, e.g., the remote estimator can use more energy or can implement sophisticated channel coding. But interestingly, imperfect feedback links can also be readily incorporated  into our framework.

Recall that at each discrete time instant $k$, the remote estimator feeds back the values $(\nu_{1,k},\dots,\nu_{M,k})$ to notify which sensors should transmit, with at most one $\nu_{m,k}=1$ in order to avoid collisions. If the feedback command is lost, then the sensor $\breve{m}$ that may have been scheduled to transmit at time $k$ will no longer do so, while the other sensors not scheduled to transmit still remain silent. Thus, from an estimation perspective, a dropout in the feedback signal is equivalent to a dropout in the forward link from the sensor to the remote estimator.  
Assume that the feedback link from the remote estimator to sensor $m$ is an i.i.d. packet dropping link with packet reception probability $\lambda_{m}^{fb}, m=1,\dots,M$, with the packet drops occurring independently of the forward links from the sensors to the remote estimator. Then for the sensor $\breve{m}$ that is scheduled to transmit, the situation is mathematically equivalent to this sensor transmitting successfully with probability $\lambda_{\breve{m}} \lambda_{\breve{m}}^{fb}$. Thus, the case of imperfect feedback links can be modelled as the case of perfect feedback links with  lower packet reception probabilities $\lambda_m \lambda_{m}^{fb}, m=1,\dots,M$. 

\section{Optimization of transmission scheduling}
\label{optimization_prob_sec}
In this section we will formulate optimization problems for determining the transmission schedules, that minimize a convex combination of the expected error covariance and expected energy usage, and  describe some numerical techniques for solving them. Structural properties of the optimal solutions to these problems will then be derived in Section \ref{structural_properties_sec}. 

Define the countable set
\begin{equation}
\label{S_defn_multi_sensor}
\mathcal{S} \triangleq \{f^n(P_{m,k|k}^s) | m=1,\dots,M, n=0,1,\dots, k=1,2,\dots \},
\end{equation}
where $f^n(.)$ is the $n$-fold composition of $f(.)$, with the convention that $f^0(X) = X$. 
Then it is clear from (\ref{remote_estimator_eqns_multi_sensor}) that $\mathcal{S}$ consists of all possible values of $\tilde{P}_{k|k}$ at the remote estimator. 
Note that if the local Kalman filters are operating in steady state, then $\mathcal{S}$ simplifies to  
\begin{equation}
\label{S_steady_state}
\begin{split}
\mathcal{S} &= \{\bar{P}_1, f(\bar{P}_1), f^2(\bar{P}_1),  \dots, \dots, \bar{P}_M, f(\bar{P}_M), f^2(\bar{P}_M), \dots \}.
\end{split}
\end{equation}
As foreshadowed in Section \ref{model_sec}, we will consider transmission policies where $\nu_{m,k}(\tilde{P}_{k-1|k-1}), m=1,\dots,M$  depends only on $\tilde{P}_{k-1|k-1}$, similar to \cite{TrimpeDAndrea_journal}. From the way in which the error covariances at the remote estimator are updated, see (\ref{remote_estimator_eqns_multi_sensor}), such policies will not depend on $x_k$, cf. \cite{WuJiaJohanssonShi}. 
To take into account energy usage, we will assume a  transmission cost of $E_{m}$ for each scheduled transmission from sensor $m$ (i.e., when $\nu_{m,k}=1$).\footnote{The transmission cost $E_m$ could represent the energy use in each transmission, but can also be regarded as a tuning parameter to provide some control on how often different sensors will transmit, e.g. increasing $E_m$ will make sensor $m$ less likely to transmit.} We will consider the following finite horizon (of horizon $K$) optimization problem:
\begin{equation}
\label{finite_horizon_problem_multi_sensor}
\begin{split}
& \min_{\{(\nu_{1,k},\dots,\nu_{M,k})\}} \sum_{k=1}^{K}  \mathbb{E} \bigg[ \beta \textrm{tr} \tilde{P}_{k|k} + (1-\beta) \sum_{m=1}^M \nu_{m,k} E_m   \bigg] \\ &
=\min_{\{(\nu_{1,k},\dots,\nu_{M,k})\}} \sum_{k=1}^{K}  \mathbb{E} \bigg[\mathbb{E} \bigg[ \beta \textrm{tr} \tilde{P}_{k|k} + (1-\beta) \sum_{m=1}^M \nu_{m,k} E_m  \\ & \quad\quad\quad\quad\quad\quad\quad\quad\quad\quad\quad\quad\bigg| \tilde{P}_{0|0},  \mathcal{I}_{k-1}, \nu_{1,k},\dots,\nu_{M,k}  \bigg] \bigg] \\
&= \min_{\{(\nu_{1,k},\dots,\nu_{M,k})\}} \sum_{k=1}^{K}  \mathbb{E}\bigg[ \mathbb{E}\bigg[  \beta \textrm{tr} \tilde{P}_{k|k} + (1-\beta) \sum_{m=1}^M \nu_{m,k} E_m  \\ & \quad\quad\quad\quad\quad\quad\quad\quad\quad\quad\quad\quad\bigg|\tilde{P}_{k-1|k-1},\nu_{1,k},\dots,\nu_{M,k}   \bigg] \bigg]
 \end{split}
 \end{equation}
for some design parameter $\beta \in (0,1)$, where the last line holds since  $\tilde{P}_{k-1|k-1}$ is a deterministic function of $\tilde{P}_{0|0}$ and $\mathcal{I}_{k-1}$, and $\tilde{P}_{k|k}$ is a function of $\tilde{P}_{k-1|k-1}, \nu_{1,k},\dots,\nu_{M,k}$, and $ \gamma_{1,k},\dots,\gamma_{M,k}$. Problem (\ref{finite_horizon_problem_multi_sensor}) minimizes a convex combination of the trace of the expected error covariance at the remote estimator and the expected sum of transmission energies of the sensors.  Due to collisions when more than one sensor is scheduled to transmit, we have
\begin{align*}
& \mathbb{E}[ \textrm{tr} \tilde{P}_{k|k} | \tilde{P}_{k-1|k-1},\nu_{1,k},\dots,\nu_{M,k}] \\ & \!\!\! = \!\sum_{m=1}^M \!\nu_{m,k} \!\prod_{n \neq m} \!(1\!-\!\nu_{n,k}) \! \left[\lambda_m \textrm{tr} P_{m,k|k}^s  \! +\! (1\!-\!\lambda_m) \textrm{tr} f(\tilde{P}_{k-1|k-1})\right]  \\ &\quad  
\!+\! \bigg(1\!-\!\sum_{m=1}^M \nu_{m,k} \prod_{n \neq m} (1\!-\!\nu_{n,k}) \bigg) \textrm{tr} f(\tilde{P}_{k-1|k-1}) \\
& = \sum_{m=1}^M \nu_{m,k} \prod_{n \neq m} (1-\nu_{n,k})  \lambda_m \textrm{tr} P_{m,k|k}^s
\\ &\quad 
+ \bigg(1-\sum_{m=1}^M \nu_{m,k} \prod_{n \neq m} (1-\nu_{n,k}) \lambda_m \bigg) \textrm{tr} f(\tilde{P}_{k-1|k-1})\\
\end{align*}
where $f(.)$ is defined in (\ref{f_defn}).

Let the functions $J_k(.): \mathcal{S} \rightarrow \mathbb{R}$ be defined recursively as: 
 \begin{align}
\label{J_fn_multi_sensor}
& J_{K+1}(\tilde{P})  =0 \nonumber \\
& J_k(\tilde{P})  = \min_{(\nu_{1},\dots,\nu_{M})} \Bigg\{ \beta \Big[ \sum_{m=1}^M \nu_{m} \prod_{n \neq m} (1-\nu_{n})  \lambda_m \textrm{tr} P_{m,k|k}^s   \nonumber \\ & \, + \Big(1-\sum_{m=1}^M \nu_{m} \prod_{n \neq m} (1-\nu_{n}) \lambda_m \Big) \textrm{tr} f(\tilde{P}) \Big] \nonumber \\ &\, + (1\!-\!\beta) \sum_{m=1}^M \nu_{m} E_m  \!+\! \sum_{m=1}^M \nu_{m} \prod_{n \neq m} (1\!-\!\nu_{n})  \lambda_m  J_{k+1}(P_{m,k|k}^s) \nonumber \\ &\, + \Big(1-\sum_{m=1}^M \nu_{m} \prod_{n \neq m} (1-\nu_{n}) \lambda_m \Big) J_{k+1}(f(\tilde{P})) \Bigg\}, \nonumber \\ & \quad \quad k=K,K-1,\dots,1.
\end{align}

Problem (\ref{finite_horizon_problem_multi_sensor}) can then solved using the dynamic programming algorithm by computing $J_k(\tilde{P}_{k-1|k-1})$ for $k=K,K-1,\dots,1$, providing the optimal $(\nu_{1,k}^*,\dots,\nu_{M,k}^*) = \textrm{argmin} J_k(\tilde{P}_{k-1|k-1}) $. 
Further call $ \mathbf{e}_0  \triangleq (0,0,\dots,0)$, $\mathbf{e}_1 \triangleq (1,0,\dots,0), \mathbf{e}_2 \triangleq (0,1,0,\dots,0), \dots, \mathbf{e}_M \triangleq (0,\dots,0,1)$,  and  
\begin{equation}
\label{V_defn}
\mathcal{V} \triangleq \{\mathbf{e}_0, \mathbf{e}_1,\dots, \mathbf{e}_M \}.
\end{equation} 
Then it is clear that the minimization in (\ref{J_fn_multi_sensor}) can be carried out over the set $\mathcal{V}$ (with cardinality $M+1$) instead of the larger set $\{0,1\}^M$ (with cardinality $2^M$). 

Note that the finite horizon problem (\ref{finite_horizon_problem_multi_sensor}) can be solved exactly via explicit enumeration, since for a given initial $\tilde{P}_{0|0}$, the number of possible values for $\tilde{P}_{k|k}, k=1,\dots,K$, is finite. When the problem has been solved (which only needs to be done once and  offline), a ``lookup table''  will be constructed at the remote estimator which allows for the transmit decisions $\nu_{m,k}$ (for different error covariances) to be easily determined in real time. 

We will also consider the infinite horizon problem:
\begin{equation}
\label{infinite_horizon_problem_multi_sensor}
\begin{split}
& \min_{\{(\nu_{1,k},\dots,\nu_{M,k})\}} \limsup_{K\rightarrow \infty}  \frac{1}{K} \sum_{k=1}^{K}  \mathbb{E} \bigg[ \mathbb{E} \bigg[  \beta \textrm{tr} \tilde{P}_{k|k} + (1-\beta) \\ & \quad\quad\quad\quad \times \sum_{m=1}^M \nu_{m,k} E_m \bigg|\tilde{P}_{k-1|k-1},\nu_{1,k},\dots,\nu_{M,k}   \bigg] \bigg]
\end{split}
 \end{equation}
 where we now assume that the local Kalman filters are operating in the steady state regime, with $P_{m,k|k}^s = \bar{P}_m, \forall k$.
Problem (\ref{infinite_horizon_problem_multi_sensor}) is a Markov decision process (MDP) based stochastic control problem with $(\nu_{1,k},\dots,\nu_{M,k})$ as the ``action'' and $\tilde{P}_{k-1|k-1}$ as the ``state'' at time $k$.\footnote{In (\ref{infinite_horizon_problem_multi_sensor}), ``limsup'' is used instead of ``lim'' since in some MDPs the limit may not exist. However, if the conditions of Theorem \ref{Bellman_eqn_lemma_multi_sensor} are satisfied then the limit will exist.}
 The Bellman equation for problem (\ref{infinite_horizon_problem_multi_sensor}) is 
 \begin{align}
 \label{Bellman_eqn_multi_sensor}
& \rho + h(\tilde{P})  = \min_{(\nu_{1},\dots,\nu_{M}) \in \mathcal{V}} \Bigg\{ \beta \Big[ \sum_{m=1}^M \nu_{m} \prod_{n \neq m} (1-\nu_{n})  \lambda_m \textrm{tr} \bar{P}_m  \nonumber \\ & \quad\quad + \Big(1-\sum_{m=1}^M \nu_{m} \prod_{n \neq m} (1-\nu_{n}) \lambda_m \Big) \textrm{tr} f(\tilde{P}) \Big] \nonumber \\ & \quad\quad+ (1-\beta) \sum_{m=1}^M \nu_{m} E_m  +\sum_{m=1}^M \nu_{m} \prod_{n \neq m} (1-\nu_{n})  \lambda_m  h(\bar{P}_m)  \nonumber \\ & \quad \quad + \Big(1-\sum_{m=1}^M \nu_{m} \prod_{n \neq m} (1-\nu_{n}) \lambda_m \Big) h(f(\tilde{P})) \Bigg\}
 \end{align}
where $\rho$ is the optimal average cost per stage and $h(.)$ is the differential cost or relative value function \cite[pp.388-389]{Bertsekas_DP1}. For the infinite horizon problem (\ref{infinite_horizon_problem_multi_sensor}), existence of optimal stationary policies can be ensured via the following result:
\begin{theorem} 
\label{Bellman_eqn_lemma_multi_sensor}
Suppose that $A$ is either (i) stable, or (ii) unstable but with  $\lambda_m > 1 - \frac{1}{\max_{i} |\sigma_i(A)|^2}$ for at least one $m \in \{1,\dots,M\}$, where $\sigma_i(A)$ is an eigenvalue of $A$.  Then there exist a constant $\rho$ and a function $h(.)$ satisfying the Bellman equation (\ref{Bellman_eqn_multi_sensor}).  
\end{theorem}
 
\begin{proof}
See Appendix \ref{Bellman_eqn_lemma_multi_sensor_proof}.
\end{proof}

\begin{remark}
In the case of a single sensor and unstable $A$, the condition $\lambda_1 > 1 - \frac{1}{\max_{i} |\sigma_i(A)|^2}$ in Theorem \ref{Bellman_eqn_lemma_multi_sensor} corresponds to the necessary and sufficient condition for estimator stability when the sensor transmits local estimates over an i.i.d. packet dropping link, see\cite{Schenato,XuHespanha}. 
\end{remark}

\begin{remark}
Dynamic programming techniques have also been used to design event triggered estimation schemes in, e.g., \cite{ImerBasar,CogillLallHespanha,LiLemmonWang}. However, these works assume a priori that the transmission policy is of threshold-type, whereas here we don't make this assumption but instead prove in Section \ref{structural_properties_sec} that the optimal policy is of threshold-type. 
\end{remark}

As a consequence of Theorem \ref{Bellman_eqn_lemma_multi_sensor}, Problem (\ref{infinite_horizon_problem_multi_sensor}) can be solved using methods such as the relative value iteration algorithm \cite[p.391]{Bertsekas_DP1}. In computations, since the state space is (countably) infinite, one can first truncate $\mathcal{S}$ in (\ref{S_steady_state}) to 
\begin{equation}
\label{S_truncated}
\begin{split}
\mathcal{S}^N &\triangleq \{\bar{P}_1, f(\bar{P}_1),  \dots, f^{N-1}(\bar{P}_1), \bar{P}_2, f(\bar{P}_2),  \dots, f^{N-1}(\bar{P}_2),\\ & \quad\quad \dots,\dots, \bar{P}_M, f(\bar{P}_M),  \dots,f^{N-1}(\bar{P}_M) \},
\end{split}
\end{equation}
which will cover all possible error covariances with up to $N-1$ successive packet drops or non-transmissions. 
We then use the relative value iteration algorithm to solve the resulting finite state space MDP problem, as follows: 
For a given $N$, define for $l=0,1,2,\dots$ the value functions $V_l(.): \mathcal{S}^N \rightarrow \mathbb{R}$ by:
\begin{equation*}
\begin{split}
& V_{l+1} (\tilde{P})  \triangleq \min_{(\nu_{1},\dots,\nu_{M}) \in \mathcal{V}} \Bigg\{ \beta \Big[ \sum_{m=1}^M \nu_{m} \prod_{n \neq m} (1-\nu_{n})  \lambda_m \textrm{tr} \bar{P}_m  \\ & \quad \quad + \Big(1-\sum_{m=1}^M \nu_{m} \prod_{n \neq m} (1-\nu_{n}) \lambda_m \Big) \textrm{tr} f(\tilde{P}) \Big] \\ &\quad\quad + (1-\beta) \sum_{m=1}^M \nu_{m} E_m  +\sum_{m=1}^M \nu_{m} \prod_{n \neq m} (1-\nu_{n})  \lambda_m  V_l(\bar{P}_m)  \\ & \quad\quad + \Big(1-\sum_{m=1}^M \nu_{m} \prod_{n \neq m} (1-\nu_{n}) \lambda_m \Big) V_l(f(\tilde{P})) \Bigg\}.
\end{split}
\end{equation*}
Let $\tilde{P}_f \in \mathcal{S}^N$ be a fixed state (which can be chosen arbitrarily). The relative  value iteration algorithm  is then given by computing:
\begin{equation}
\label{relative_value_iteration}
\begin{split}
& h_{l+1}(\tilde{P})  \triangleq V_{l+1}(\tilde{P}) - V_{l+1}(\tilde{P}_f) 
\end{split}
\end{equation}
for $l=0,1,2,\dots$. 
As $l \rightarrow \infty$, we have $h_l(\tilde{P}) \rightarrow h(\tilde{P}), \forall \tilde{P} \in \mathcal{S}^N$, with $h(.)$ satisfying the Bellman equation (\ref{Bellman_eqn_multi_sensor}). In practice, the algorithm $(\ref{relative_value_iteration})$ terminates once the differences $h_{l+1}(\tilde{P})-h_l(\tilde{P})$ become smaller than a desired level of accuracy $\varepsilon$. 
One then compares the solutions obtained as $N$ increases to determine an appropriate value of $N$ for truncation of the state space $\mathcal{S}$, see Chapter~8 of \cite{Sennott_book} for further details.

\section{Structural properties of optimal transmission scheduling}
\label{structural_properties_sec}
Numerical solutions to the optimization problems (\ref{finite_horizon_problem_multi_sensor}) and (\ref{infinite_horizon_problem_multi_sensor}) via dynamic programming or solving MDPs do not provide much insight into the form of the optimal solution. In this section, we will derive some structural results on the optimal solutions to the finite horizon problem (\ref{finite_horizon_problem_multi_sensor}) and the infinite horizon problem (\ref{infinite_horizon_problem_multi_sensor}) in  Sections \ref{finite_horizon_structural_sec} and \ref{infinite_horizon_structural_sec} respectively. To be more specific,  we will prove that if the error covariance is large, then a sensor will always be scheduled to transmit (for unstable $A$), and show threshold-type behaviour in switching from one sensor to another. In Section \ref{single_sensor_sec}, we specialize these results to demonstrate that, in the case of a single sensor, a threshold policy is optimal, and derive simple analytical expressions for the expected energy usage and expected error covariance.

\emph{Preliminaries}: For symmetric matrices $X$ and $Y$, we say that $X\leq Y$ if $Y-X$ is positive semi-definite, and $X < Y$ if $Y-X$ is positive definite. In general, $``\leq"$ only gives a partial ordering on the set $\mathcal{S}$ defined in (\ref{S_defn_multi_sensor}).
Let $\bf{S}$ denote the set of all positive semi-definite matrices. In this section, we will say that a function $F(.): \bf{S} \rightarrow \mathbb{R}$ is \emph{increasing} if 
\begin{equation}
\label{increasing_fn_defn}
X \leq Y \Rightarrow F(X) \leq F(Y).
\end{equation}
Note that (\ref{increasing_fn_defn}) does not take into account the situations where neither $X \leq Y$ nor $Y \leq X$ holds under the partial order $``\leq"$. 

\begin{lemma}
\label{f_increasing_lemma}
The function $\textnormal{tr} f(X) =\textnormal{tr}( A X A^T +Q)$ is an increasing function of $X$. 
\end{lemma}
\begin{proof}
This is easily seen from the definition.
\end{proof}

\subsection{Finite Horizon Costs}
\label{finite_horizon_structural_sec}

\begin{lemma}
\label{J_fn_increasing_lemma_multi_sensor}
The functions $J_k(\tilde{P})$ defined in (\ref{J_fn_multi_sensor}) are increasing functions of $\tilde{P}$.
\end{lemma}

\begin{proof}
The proof is by induction. The case of $J_{K+1}(.)$ is clear. Now assume that $J_{K+1}(\tilde{P}),J_{K}(\tilde{P}),\dots,J_{k+1}(\tilde{P})$ are increasing functions of $\tilde{P}$. Then
$J_k(\tilde{P})$ given in (\ref{J_fn_multi_sensor}) is increasing in $\tilde{P}$ by Lemma \ref{f_increasing_lemma} and the induction hypothesis, noting that $\Big(1-\sum_{m=1}^M \nu_{m} \prod_{n \neq m} (1-\nu_{n}) \lambda_m \Big) \geq 0$. 
\end{proof}

Since the minimization in  (\ref{J_fn_multi_sensor})  is over the set $\mathcal{V}$ given in (\ref{V_defn}), $J_k(\tilde{P})$ can also be expressed as:
\begin{align}
\label{J_k_alternative}
J_k(\tilde{P})  & = \min \Big\{  \beta[ \lambda_1 \textrm{tr}P^s_{1,k|k} + (1-\lambda_1) \textrm{tr}f(\tilde{P})] + (1-\beta) E_1
\nonumber \\ & \quad \quad + \lambda_1 J_{k+1}(P^s_{1,k|k}) + (1-\lambda_1) J_{k+1} (f(\tilde{P})), 
\nonumber  \\ & \quad \quad \quad \vdots 
\nonumber \\& \quad \beta[ \lambda_M \textrm{tr}P^s_{M,k|k} + (1-\lambda_M) \textrm{tr}f(\tilde{P})] + (1-\beta) E_M \nonumber  \\ & \quad \quad + \lambda_M J_{k+1}(P^s_{M,k|k}) + (1-\lambda_M) J_{k+1} ( f(\tilde{P})), 
\nonumber \\ & \quad \beta \textrm{tr}f(\tilde{P}) +  J_{k+1} ( f(\tilde{P})) \Big\}.  
\end{align}

\begin{theorem}
\label{multi_sensor_structural_properties_theorem}
(i) The functions defined by
\begin{equation*}
\begin{split} 
& \phi_{m,k}(\tilde{P}) \triangleq \beta \textnormal{tr}f(\tilde{P})\! +\!  J_{k+1} (f(\tilde{P}))  \!-\! \beta[ \lambda_m \textnormal{tr} P^s_{m,k|k} \!+\! (1\!-\!\lambda_m) \\ & \times \textnormal{tr}f(\tilde{P})]  \!-\! (1\!-\!\beta) E_m   \!-\! \lambda_m J_{k+1}(P^s_{m,k|k})\! -\! (1\!-\!\lambda_m) J_{k+1} ( f(\tilde{P} )) 
\end{split}
\end{equation*}
for $ m=1,\dots,M$, $k=1,\dots,K$, are increasing functions of $\tilde{P}$. 
\\(ii) Define 
\begin{equation*}
\begin{split}
 \psi_{m,k}(\tilde{P}) & \triangleq \beta[ \lambda_m \textnormal{tr}P^s_{m,k|k} + (1-\lambda_m) \textnormal{tr}f( \tilde{P} )] + (1-\beta) E_m \\ & \quad \quad + \lambda_m J_{k+1}(P^s_{m,k|k}) + (1-\lambda_m) J_{k+1} ( f(\tilde{P} )) 
\end{split}
\end{equation*}
for $ m=1,\dots,M$, $k=1,\dots,K$.
Suppose that for some $m, n \in \{1,\dots,M\}$, and $\tilde{P}, \tilde{P}' \in \mathcal{S}$ with $\tilde{P}' \geq \tilde{P}$, we have
\begin{equation}
\label{inequality1-2}
 \psi_{m,k}(\tilde{P}) \leq \psi_{n,k}(\tilde{P}) \textrm{ and } \psi_{m,k}(\tilde{P}') \geq \psi_{n,k}(\tilde{P}').
\end{equation}
Then for $\tilde{P}'' \geq \tilde{P}'$, we have 
$
\psi_{m,k}(\tilde{P}'') \geq \psi_{n,k}(\tilde{P}'').
$
\end{theorem}

\begin{proof}
(i) We can simplify the functions to
\begin{equation}
\label{phi_simplified}
\begin{split}
\phi_{m,k} (\tilde{P}) & =\beta \lambda_m \textrm{tr} f(\tilde{P} ) + \lambda_m J_{k+1} (f(\tilde{P} )) - [\beta \lambda_m \textrm{tr}P^s_{m,k|k}\\ & \quad \quad + (1-\beta) E_m  + \lambda_m J_{k+1} (P^s_{m,k|k})]
 \end{split}
 \end{equation}
which are increasing in $\tilde{P}$ by Lemmas \ref{f_increasing_lemma} and \ref{J_fn_increasing_lemma_multi_sensor}. 
\\(ii) Rewrite (\ref{inequality1-2})  as 
\begin{equation}
\label{inequality3}
\begin{split}
&(1-\lambda_m) [ \beta \textrm{tr} f(\tilde{P}) + J_{k+1} (f( \tilde{P}))] \\ & \quad + \beta \lambda_m \textrm{tr} P^s_{m,k|k} + (1-\beta) E_m + \lambda_m J_{k+1} (P^s_{m,k|k}) \\ 
& \leq (1-\lambda_n) [ \beta \textrm{tr} f(\tilde{P} ) + J_{k+1} (f( \tilde{P}))] \\ & \quad  + \beta \lambda_n \textrm{tr}P^s_{n,k|k} + (1-\beta) E_n + \lambda_n J_{k+1} (P^s_{n,k|k}) 
\end{split}
\end{equation}
and 
\begin{equation}
\label{inequality4}
\begin{split}
&(1-\lambda_m) [ \beta \textrm{tr} f( \tilde{P}') + J_{k+1} (f( \tilde{P}' ))] \\ & \quad + \beta \lambda_m \textrm{tr}P^s_{m,k|k} + (1-\beta) E_m + \lambda_m J_{k+1} (P^s_{m,k|k}) \\ 
& \geq (1-\lambda_n) [ \beta \textrm{tr} f( \tilde{P}' ) + J_{k+1} (f( \tilde{P}'))] \\ & \quad + \beta \lambda_n \textrm{tr} P^s_{n,k|k} + (1-\beta) E_n + \lambda_n J_{k+1} (P^s_{n,k|k}). 
\end{split}
\end{equation}
Since $\tilde{P}' \geq \tilde{P}$, expressions (\ref{inequality3})-(\ref{inequality4}) and Lemmas \ref{f_increasing_lemma} and \ref{J_fn_increasing_lemma_multi_sensor} imply that 
$\lambda_m \leq \lambda_n$. 
Thus, for $\tilde{P}'' \geq \tilde{P}'$ we have 
\begin{equation*}
\begin{split}
&(1-\lambda_m) [ \beta \textrm{tr} f( \tilde{P}'' ) + J_{k+1} (f( \tilde{P}''))] \\ & \quad + \beta \lambda_m \textrm{tr}P^s_{m,k|k} + (1-\beta) E_m + \lambda_m J_{k+1} (P^s_{m,k|k}) \\ 
& \geq (1-\lambda_n) [ \beta \textrm{tr} f(\tilde{P}'') + J_{k+1} (f( \tilde{P}'' ))] \\ & \quad + \beta \lambda_n \textrm{tr} P^s_{n,k|k} + (1-\beta) E_n + \lambda_n J_{k+1} (P^s_{n,k|k}). 
\end{split}
\end{equation*}
\end{proof}

Theorem \ref{multi_sensor_structural_properties_theorem} characterizes some structural properties of the optimal transmission schedule over a finite horizon. Theorem \ref{multi_sensor_structural_properties_theorem}(i) and expression (\ref{J_k_alternative}) allow one to conclude that for unstable $A$ and sufficiently large $\tilde{P}$ one will always schedule a sensor to transmit. This is because $\textrm{tr} f(\tilde{P}) \rightarrow \infty$ as $\tilde{P}$ increases, so that (\ref{phi_simplified})
is always positive for sufficiently large $\tilde{P}$. On the other hand, for stable $A$, we could encounter the situation where sensors are never scheduled to transmit if the costs of transmission $E_m$ are large, since now $\textrm{tr} f(\tilde{P})$ is always bounded (where the bound could depend on the initial error covariance). 

Theorem \ref{multi_sensor_structural_properties_theorem}(ii) and expression (\ref{J_k_alternative}) further show that the optimal schedule exhibits threshold-type behaviour in switching from one sensor to another: If for some $\tilde{P}$, sensor $m$ is scheduled to transmit, while for some larger $\tilde{P}'$, sensor $n$ (with $n \neq m$) is scheduled to transmit, then sensor $m$ will not transmit $\forall \tilde{P}'' > \tilde{P}'$.  Note however that Theorem \ref{multi_sensor_structural_properties_theorem} may not cover all possible situations, since the set $\mathcal{S}$ given by (\ref{S_defn_multi_sensor}) is in general not a totally ordered set.

For  scalar systems (or systems with scalar states $x_k$ and hence scalar $\tilde{P}_{k|k}$), the set $\mathcal{S}$ \emph{is} totally ordered, and Theorem \ref{multi_sensor_structural_properties_theorem} and (\ref{J_k_alternative}) can be used to provide a fairly complete characterization.\footnote{The set $\mathcal{S}$ is also totally ordered in the vector system, single sensor situation in steady state, see Section \ref{single_sensor_sec}.} For example, in the situation with two sensors, we have:
\begin{corollary}
\label{two_sensor_structural_lemma} 
For a scalar system with two sensors, for each $k \in \{1,\dots,K\}$, the behaviour of the optimal $\nu_{1,k}^*$ and $\nu_{2,k}^*$ falls into exactly one of the following four scenarios:
\\(i) There exists a $\tilde{P}_{1,k-1}^{\textrm{th}}$ such that $\nu_{2,k}^*=0, \forall \tilde{P}_{k-1|k-1}$,  $\nu_{1,k}^* = 0$ for $\tilde{P}_{k-1|k-1} < \tilde{P}_{1,k-1}^{\textrm{th}}$, and $\nu_{1,k}^* = 1$ for $\tilde{P}_{k-1|k-1} \geq \tilde{P}_{1,k-1}^{\textrm{th}}$.
\\(ii) There exists a $\tilde{P}_{2,k-1}^{\textrm{th}}$ such that $\nu_{1,k}^*=0, \forall \tilde{P}_{k-1|k-1}$, $\nu_{2,k}^* = 0$ for $\tilde{P}_{k-1|k-1} < \tilde{P}_{2,k-1}^{\textrm{th}}$, and $\nu_{2,k}^* = 1$ for $\tilde{P}_{k-1|k-1} \geq \tilde{P}_{2,k-1}^{\textrm{th}}$.
\\(iii) There exists some $\tilde{P}_{1,k-1}^{\textrm{th}}$ and $\tilde{P}_{2,k-1}^{\textrm{th}}$ such that $\nu_{2,k}^* = 0$ for $\tilde{P}_{k-1|k-1} < \tilde{P}_{2,k-1}^{\textrm{th}}$, $\nu_{2,k}^* = 1$ for $\tilde{P}_{2,k-1}^{\textrm{th}} \leq \tilde{P}_{k-1|k-1} < \tilde{P}_{1,k-1}^{\textrm{th}}$, and $\nu_{1,k}^* = 1$ for $\tilde{P}_{k-1|k-1} \geq \tilde{P}_{1,k-1}^{\textrm{th}}$.
\\(iv) There exists some $\tilde{P}_{1,k-1}^{\textrm{th}}$ and $\tilde{P}_{2,k-1}^{\textrm{th}}$ such that $\nu_{1,k}^* = 0$ for $\tilde{P}_{k-1|k-1} < \tilde{P}_{1,k-1}^{\textrm{th}}$, $\nu_{1,k}^* = 1$ for $\tilde{P}_{1,k-1}^{\textrm{th}} \leq \tilde{P}_{k-1|k-1} < \tilde{P}_{2,k-1}^{\textrm{th}}$, and $\nu_{2,k}^* = 1$ for $\tilde{P}_{k-1|k-1} \geq \tilde{P}_{2,k-1}^{\textrm{th}}$.
\end{corollary}

From numerical simulations, one finds that each of the above four scenarios can occur (for different parameter values), see Section \ref{numerical_multi_sensor_sec}.

\subsection{Infinite Horizon Costs}
\label{infinite_horizon_structural_sec}
For the infinite horizon problem (\ref{infinite_horizon_problem_multi_sensor}), we have the following counterpart to Theorem \ref{multi_sensor_structural_properties_theorem}.

\begin{lemma}
\label{multi_sensor_structural_properties_theorem_inf_horizon}
(i) The functions defined by
\begin{equation*}
\begin{split} 
\phi_m (\tilde{P}) &\triangleq  \beta \textnormal{tr}f(\tilde{P}) +  h (f(\tilde{P}))  - \beta[ \lambda_m \textnormal{tr}\bar{P}_m + (1-\lambda_m) \textnormal{tr}f(\tilde{P})] \\ & \quad \quad - (1-\beta) E_m   - \lambda_m h(\bar{P}_m) - (1-\lambda_m)h ( f(\tilde{P} )) 
\end{split}
\end{equation*}
for $m=1,\dots,M$, are increasing functions of $\tilde{P}$. 
\\(ii) Define 
\begin{equation*}
\begin{split}
\psi_m(\tilde{P}) & \triangleq \beta[ \lambda_m \textnormal{tr}\bar{P}_m + (1-\lambda_m) \textnormal{tr}f( \tilde{P} )] + (1-\beta) E_m \\ & \quad + \lambda_m h(\bar{P}_m) + (1-\lambda_m) h ( f(\tilde{P} ))  
\end{split}
\end{equation*}
for $m=1,\dots,M$. 
Suppose that for some $m, n \in \{1,\dots,M\}$, and $\tilde{P}, \tilde{P}' \in \mathcal{S}$ with $\tilde{P}' \geq \tilde{P}$, we have
$\psi_m(\tilde{P}) \leq \psi_n(\tilde{P})$
and
$\psi_m(\tilde{P'}) \geq \psi_n(\tilde{P'}).$
Then for $\tilde{P}'' \geq \tilde{P}'$, we have 
$\psi_m(\tilde{P''}) \geq \psi_n(\tilde{P''}).$
\end{lemma}

\begin{proof}
Recalling the relative value iteration algorithm (\ref{relative_value_iteration}), one can show using similar arguments as in the proof of Theorem \ref{multi_sensor_structural_properties_theorem}, that the  properties in Theorem \ref{multi_sensor_structural_properties_theorem} also hold when $J_{k+1}(.)$ is replaced with $h_l(.)$. Since $h_l(\tilde{P}) \rightarrow h(\tilde{P})$ as $ l \rightarrow \infty$, the result follows. 
\end{proof}

In the infinite horizon situation, any thresholds (which for the finite horizon situation are generally time-varying) become constant, i.e. do not depend on $k$. Thus for example, with the  scalar system, two sensor situation considered in Corollary \ref{two_sensor_structural_lemma}, one may replace $\tilde{P}_{1,k-1}^{\textrm{th}}$ and 
$\tilde{P}_{2,k-1}^{\textrm{th}}$ with $\tilde{P}_{1}^{\textrm{th}}$ and 
$\tilde{P}_{2}^{\textrm{th}}$ respectively, see also Theorem \ref{threshold_policy_theorem}.

 \begin{remark}
 \label{computational_remark}
The structural results derived above allow for significant reductions in the amount of computation required to solve problems (\ref{finite_horizon_problem_multi_sensor}) and (\ref{infinite_horizon_problem_multi_sensor}). For example, by Theorem \ref{multi_sensor_structural_properties_theorem} or  \ref{multi_sensor_structural_properties_theorem_inf_horizon}, if for some $P$ one has $\nu_m^*=1$, and for a larger $P'$ one has $\nu_m^*=0$, then one can automatically set $\nu_m^*=0$ for all $P'' \geq P'$. See also \cite{NgoKrishnamurthy} for a related discussion.

When the covariance matrices are not comparable in the positive semi-definite ordering, then the full dynamic programming or value iteration algorithm will need to be run in order to solve the optimization problems. Nevertheless, when the remote estimator (\ref{remote_estimator_eqns_multi_sensor}) is used the computational complexity  is not prohibitive.  In the case where the local Kalman filters have converged to steady state, which is likely when the horizon $K$ is large or if we're interested in the infinite horizon, the  ``state space'' $\mathcal{S}$ simplifies to 
(\ref{S_steady_state}),
which in numerical approaches is truncated to the set $\mathcal{S}^N$ defined in
(\ref{S_truncated}).  The cardinality of  $\mathcal{S}^N $ is $N  M$, which is not exponential in the number of sensors $M$ or the horizon $K$. Furthermore, the ``action space'' $\mathcal{V}$ defined in (\ref{V_defn}) has cardinality $M+1$, which is also linear in $M$.  
 \end{remark}

\subsection{Single Sensor Case}
\label{single_sensor_sec}
In this subsection we will focus on a vector system with a single sensor, and where the local Kalman filter operates in steady state, to further characterize the optimal solutions to problems (\ref{finite_horizon_problem_multi_sensor}) and (\ref{infinite_horizon_problem_multi_sensor}). 
For notational simplicity, we will drop the subscript ``1'' from quantities such as $\nu_{1,k}, \bar{P}_1,  P_{1,k-1|k-1}$. 
Recall the set $\mathcal{S}$ defined by (\ref{S_defn_multi_sensor}), which in the multi-sensor case is  not totally ordered in general.  For the single sensor case in steady state, $\mathcal{S}$ becomes:
\begin{equation}
\label{S_defn_single_sensor}
 \mathcal{S} \triangleq \{\bar{P}, f(\bar{P}), f^2(\bar{P}), \dots \}.
 \end{equation}
\begin{lemma}
\label{total_ordering_lemma}
In the single sensor case, there is a total ordering on the elements of $\mathcal{S}$ given by
$$ \bar{P} \leq f(\bar{P}) \leq f^2 (\bar{P}) \leq \dots .$$
\end{lemma}
\begin{proof}
We use induction. We have that $f(\bar{P}) \geq \bar{P}$ from, e.g., \cite{ShiEpsteinMurray}. Now assume that $f^n (\bar{P}) \geq f^{n-1} (\bar{P})$. Then
$$f^{n+1}(\bar{P}) = f(f^n(\bar{P})) \geq f(f^{n-1}(\bar{P})) = f^n (\bar{P})$$
where the inequality comes from Lemma \ref{f_increasing_lemma} and the induction hypothesis. Hence,  by induction,
$$ \bar{P} \leq f(\bar{P}) \leq f^2 (\bar{P}) \leq \dots .$$
\end{proof}

Using (\ref{J_k_alternative}), Theorems \ref{multi_sensor_structural_properties_theorem}, \ref{multi_sensor_structural_properties_theorem_inf_horizon}, and Lemma \ref{total_ordering_lemma}, we then conclude the following threshold behaviour of the optimal solution:
\begin{theorem}
\label{threshold_policy_theorem}
(i) In the single sensor case, the optimal solution to the finite horizon problem (\ref{finite_horizon_problem_multi_sensor}) is of the form: 
$$\nu_k^* = \left\{\begin{array}{ccc} 0 & , & P_{k-1|k-1} < P_{k-1|k-1}^{\textrm{th}} \\ 1 & , & P_{k-1|k-1} \geq P_{k-1|k-1}^{\textrm{th}} \end{array} \right.$$
for some thresholds $P_{k-1|k-1}^{\textrm{th}} , k=1,\dots,K$, where the thresholds may be infinite (meaning that  $\nu_k^*=0, \forall P_{k-1|k-1} \in \mathcal{S}$) when $A$ is stable. 
\\ (ii) In the single sensor case, the optimal solution to the infinite horizon problem (\ref{infinite_horizon_problem_multi_sensor}) is of the form: 
\begin{equation}
\label{infinite_horizon_soln}
\nu_k^* = \left\{\begin{array}{ccc} 0 & , & P_{k-1|k-1} < P^{\textrm{th}} \\ 1 & , & P_{k-1|k-1} \geq P^{\textrm{th}} \end{array} \right.
\end{equation}
for some constant threshold $P^{\textrm{th}}$, where the threshold may be infinite  when $A$ is stable. 
\end{theorem}

\begin{remark}
In Theorem \ref{threshold_policy_theorem}, we could have $P_{k-1|k-1}^{\textrm{th}}$ or $P^{\textrm{th}} $ equal to $\bar{P}$, in which case $\nu_k^*=1, \forall P_{k-1|k-1} \in \mathcal{S}$. 
\end{remark} 

\begin{remark}
As mentioned in the Introduction, Theorem \ref{threshold_policy_theorem} was proved in our conference contribution \cite{LeongDeyQuevedo_ECC} using the theory of submodular functions. 
Under a related setup that minimizes an expected error covariance measure subject to a constraint on the communication rate, the optimality of threshold policies over an infinite horizon was also proved using different techniques in \cite{MoSinopoliShiGarone}. 
\end{remark} 

Thus in the single sensor case the optimal policy is a threshold policy on the error covariance. This also allows us to derive simple analytical expressions for the expected energy usage and expected error covariance for the single sensor case over an infinite horizon. A similar analysis can be carried out for the finite horizon situation but the expressions will be more complicated due to the thresholds $P_{k-1|k-1}^{th}$ in Theorem \ref{threshold_policy_theorem} being time-varying in general. 

Let $t \in \mathbb{N}$ be such that $f^t(\bar{P}) = P^{\textrm{th}}\in \mathcal{S}$, see (\ref{infinite_horizon_soln}). Note that $t$ will depend on the value of $\beta$ chosen in problem (\ref{infinite_horizon_problem_multi_sensor}). 
Then the evolution of the error covariance at the remote estimator can be modelled as the (infinite) Markov chain shown in  Fig. \ref{Markov_chain}, where state $i$ of the Markov chain corresponds to the value $f^i(\bar{P}), i = 0, 1, 2, \dots$, with $f^0(\bar{P}) \triangleq \bar{P}$.
\begin{figure}[htb!]
\centering 
\includegraphics[scale=0.32]{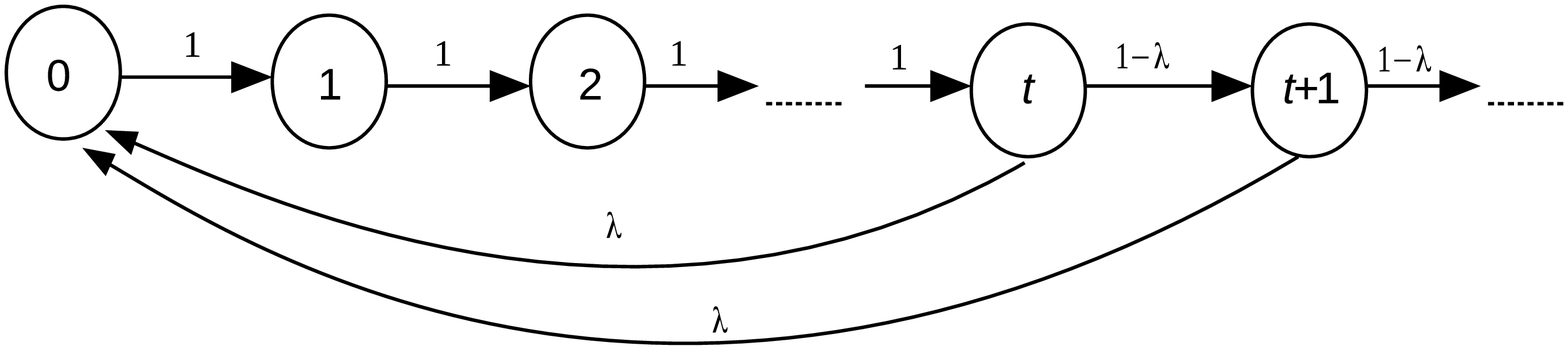} 
\caption{Markov chain for threshold policy}
\label{Markov_chain}
\end{figure}

The transition probability matrix $\mathbf{P}$ for the  Markov chain can be written as:
$$\mathbf{P} \!= \! \left[\! \begin{array}{ccccccccc}
0 & 1 & 0 & \dots & \dots & & & & \dots  \\
0 & 0 & 1 & 0 & \dots & & & & \dots \\
\vdots & & \ddots & \ddots & & & & & \\
0 & \dots & \dots & 0 & 1 & 0 & \dots & & \dots \\
\lambda & 0 & \dots & & 0 & 1-\lambda & 0 & \dots & \dots \\
\lambda & 0 & \dots & & & 0 & 1-\lambda & 0 & \dots \\
\vdots & \vdots & & & & & & & \ddots
 \end{array} \!\right].$$

For $\lambda \in (0,1)$, one can easily verify that this Markov chain is irreducible, aperiodic, and with all states being positive recurrent. 
Then the stationary distribution 
 $$\pi = \left[ \begin{array}{cccccccc} 
 \pi_0 & \pi_1 & \pi_2 & \dots & \pi_t & \pi_{t+1} & \pi_{t+2} & \dots
 \end{array}
 \right],$$
 where $\pi_j$ is the stationary probability of the Markov chain being in state $j$, exists and can be  computed using the relation $\pi=\pi \mathbf{P}$. We find after some calculations that $\pi_j=\pi_0, j=1,\dots,t$, and $\pi_j=(1-\lambda)^{j-t} \pi_0, j=t+1,t+2,\dots$, and so
 $$\pi_0 = \frac{1}{t+1/\lambda} = \frac{\lambda}{\lambda t + 1}.$$
Hence 
$$ \pi_j = \left\{\begin{array}{ccl} \frac{\lambda}{\lambda t + 1} & , & j=0,\dots,t \\
\frac{(1-\lambda)^{j-t} \lambda}{\lambda t + 1} & , & j=t+1,t+2,\dots .\end{array}   \right.$$

We can now derive analytical expressions for the expected energy usage and expected error covariance. 
 For the expected energy usage, since the sensor transmits only when the Markov chain is in states $t,t+1, \dots$, an energy amount of $E$ is used in reaching the states  corresponding to $\bar{P}, f^{t+1}(\bar{P}), f^{t+2}(\bar{P}), \dots$. Hence 
\begin{equation}
\label{expected_energy}
\begin{split}
\mathbb{E}[\textrm{energy}] & = E [ \pi_0 + \pi_{t+1} + \pi_{t+2} + \dots ]\\
& = E \pi_0 [ 1 + 1-\lambda + (1-\lambda)^2 + \dots]\\
& = \frac{E\pi_0}{\lambda} = \frac{E}{\lambda t + 1}.
\end{split}
\end{equation}
For the expected error covariance, we have
\begin{equation}
\label{expected_error_covariance}
\mathbb{E}[\textrm{tr} P_{k|k}] = \pi_0 \textrm{tr}(\bar{P}) + \pi_1 \textrm{tr} (f(\bar{P}))+\pi_2 \textrm{tr} (f^2(\bar{P})) + \dots 
\end{equation}
which can be computed numerically. Under the assumption that  $\lambda > 1 - \frac{1}{\max_{i} |\sigma_i(A)|^2}$, $\mathbb{E}[\textrm{tr} P_{k|k}]$ will be finite, by a similar argument as that used in the proof of Theorem \ref{Bellman_eqn_lemma_multi_sensor}.

\section{Transmitting measurements}
\label{tx_meas_sec}
In this section we will study the situation where sensor measurements instead of local state estimates are transmitted to the remote estimator. In particular, we wish to derive structural results on the optimal transmission schedule. An advantage with transmitting measurements is that detectability at each sensor is not required, but just the detectability of the overall system \cite{TrimpeDAndrea_journal}. In addition, local Kalman filtering at the individual sensors is not required. The optimal remote estimator when sending measurements also has a simpler form than the optimal remote estimator derived in (\ref{multi_sensor_optimal_estimator}) when sending state estimates (though not as simple as the suboptimal estimator (\ref{remote_estimator_eqns_multi_sensor})), which makes it amenable to analysis.  Our descriptions of the model and optimization problem below will be kept brief, in order to proceed quickly to the structural results. 

\subsection{System Model}
The process and measurements follow the same model as in (\ref{state_eqn})-(\ref{measurement_eqn}). 
Instead of assuming that the individual sensors are detectable, we will now merely assume that $(A,C)$ is detectable, where $C \triangleq \left[ \begin{array}{ccc} C_1^T & \dots & C_M^T \end{array} \right]^T$ is the matrix formed by stacking $C_1,\dots,C_M$ on top of each other. 

Let $\nu_{m,k} \in \{0,1\}, m=1,\dots,M$ be  decision variables such that $\nu_{m,k}=1$ if the measurement $y_{m,k}$ (rather than the local state estimate) is to be transmitted to the remote estimator at time $k$, and $\nu_{m,k}=0$ if there is no transmission. As before (see Fig. \ref{system_model}), the transmit decisions $\nu_{m,k}$ are to be decided at the remote estimator and assumed to only depend on the error covariance at the remote estimator. 

At the remote estimator, if no sensors are scheduled to transmit, then the state estimates and error covariances are updated by (\ref{remote_estimator_no_tx}). If sensor $\breve{m} \in \{1,\dots,M\}$ has been scheduled by the remote estimator to transmit at time $k$
then the state estimates and error covariances at the remote estimator
are now updated as follows:
\begin{equation}
\label{multi_sensor_optimal_estimator_tx_meas}
\begin{split}
\hat{x}_{k+1|k} & = A \hat{x}_{k|k} \\
\hat{x}_{k|k} & = \hat{x}_{k|k-1} + \gamma_{\breve{m},k} K_{\breve{m},k} (y_{\breve{m},k} - C_{\breve{m}} \hat{x}_{k|k-1}) \\
P_{k+1|k} & = A P_{k|k} A^T + Q \\
P_{k|k} & = P_{k|k-1} - \gamma_{\breve{m},k} K_{\breve{m},k} C_{\breve{m}} P_{k|k-1} 
\end{split}
\end{equation} 
where $K_{\breve{m},k} \triangleq P_{k|k-1} C_{\breve{m}}^T (C_{\breve{m}} P_{k|k-1} C_{\breve{m}}^T + R_{\breve{m}})^{-1} $.
We can thus write:
\begin{equation}
\label{remote_estimator_eqns_tx_meas}
\begin{split}
\hat{x}_{k\!+\!1|k} & \!= \!\left\{\!\!\!\!\begin{array}{ccc}  A \hat{x}_{k|k-1} & \!\!\!\!\!, \, \nu_{m,k} \gamma_{m,k} = 0 \\ A \hat{x}_{k|k\!-\!1} \!+\! A K_{m,k} (y_{m,k} \!-\! C_m \hat{x}_{k|k\!-\!1}) & \!\!\!\!\!, \, \nu_{m,k} \gamma_{m,k} = 1 \end{array}  \right. \\
P_{k\!+\!1|k} & \!=\! \left\{\begin{array}{cll} f(P_{k|k-1}) & , \,  \nu_{m,k} \gamma_{m,k} = 0 \\ g_m(P_{k|k-1}) & , \,  \nu_{m,k} \gamma_{m,k} = 1, \end{array} \right. 
\end{split}
\end{equation} 
where $f(X) \triangleq A X A^T + Q$ as before, and 
\begin{equation}
\label{gm_defn}
g_m(X) \triangleq  A X A^T - A X C_m^T (C_m X C_m^T + R_m)^{-1} C_m X A^T + Q, 
\end{equation}
for $ m=1,\dots,M$.
In (\ref{remote_estimator_eqns_tx_meas}) the recursions are given in terms of $\hat{x}_{k+1|k}$ and $P_{k+1|k}$ rather than $\hat{x}_{k|k}$ and $P_{k|k}$, since the resulting expressions are more convenient to work with. 
 
\subsection{Optimization of Transmission Scheduling}
\label{optimization_prob_sec_tx_meas}
We consider transmission policies $\nu_{m,k}(P_{k|k-1}), m=1,\dots,M$ that depend only on $P_{k|k-1}$. 
The finite horizon  optimization problem is:
\begin{equation}
\label{finite_horizon_problem_multi_sensor_tx_meas}
\begin{split}
\min_{\{(\nu_{1,k},\dots,\nu_{M,k})\}} \sum_{k=1}^{K}  \mathbb{E}\bigg[ \mathbb{E} \bigg[ \beta & \textrm{tr} P_{k+1|k} + (1-\beta) \sum_{m=1}^M \nu_{m,k} E_m \\ 
&  \bigg|P_{k|k-1},\nu_{1,k},\dots,\nu_{M,k}   \bigg]\bigg]
 \end{split}
 \end{equation}
where we can compute
\begin{equation*}
\begin{split}
& \mathbb{E}[ \textrm{tr} P_{k+1|k} | P_{k|k-1},\nu_{1,k},\dots,\nu_{M,k}] \\ 
& = \sum_{m=1}^M \nu_{m,k} \prod_{n \neq m} (1-\nu_{n,k})  \lambda_m \textrm{tr} g_m(P_{k|k-1})
\\ &\quad 
+ \bigg(1-\sum_{m=1}^M \nu_{m,k} \prod_{n \neq m} (1-\nu_{n,k}) \lambda_m \bigg) \textrm{tr} f(P_{k|k-1})\\
 \end{split}
\end{equation*}
with $f(.)$ defined in (\ref{f_defn}) and $g_m(.)$ defined in (\ref{gm_defn}).
  Let the functions $J_k(.)$ be defined as: 
 \begin{align}
\label{J_fn_multi_sensor_tx_meas}
&J_{K+1}(P)  =0 \nonumber \\
&J_k(P)  = \min_{(\nu_{1},\dots,\nu_{M}) \in \mathcal{V}} \Bigg\{ \beta \Big[ \sum_{m=1}^M \nu_{m} \prod_{n \neq m} (1-\nu_{n})  \lambda_m \textrm{tr} g_m(P)  \nonumber  \\ & \, + \Big(1-\sum_{m=1}^M \nu_{m} \prod_{n \neq m} (1-\nu_{n}) \lambda_m \Big) \textrm{tr} f(P) \Big] \nonumber \\ &\,+ (1\!-\!\beta) \sum_{m=1}^M \nu_{m} E_m \! +\!\sum_{m=1}^M \nu_{m} \prod_{n \neq m} (1\!-\!\nu_{n})  \lambda_m  J_{k+1}(g_m(P)) \nonumber \\ &\, + \Big(1-\sum_{m=1}^M \nu_{m} \prod_{n \neq m} (1-\nu_{n}) \lambda_m \Big) J_{k+1}(f(P)) \Bigg\} , \nonumber\\ & \quad\quad \quad k=K,K-1,\dots,1.
\end{align}
Problem (\ref{finite_horizon_problem_multi_sensor_tx_meas}) can be solved using the dynamic programming algorithm by computing $J_k(P_{k|k-1})$ for $k=K,K-1,\dots,1$, with the optimal $(\nu_{1,k}^*,\dots,\nu_{M,k}^*) = \textrm{argmin} J_k(P_{k|k-1})$. 

The infinite horizon problem can be formulated in a similar manner but will be omitted for brevity. 

\subsection{Structural Properties of Optimal Transmission Scheduling}
Much of this subsection is devoted to proving Theorem \ref{multi_sensor_structural_properties_thm_tx_meas}, which is the counterpart of Theorem \ref{multi_sensor_structural_properties_theorem}(i) for scalar systems, and in particular establishes the optimality of threshold policies in the single sensor, scalar case. 
However, for vector systems we will give a counterexample (Example \ref{counterexample1}) to show that, in general, the optimal policy is not a simple threshold policy. 
The counterpart of Theorem \ref{multi_sensor_structural_properties_theorem}(ii) also turns out to be false when measurements are transmitted, and we will give another counterexample (Example \ref{counterexample2}) to illustrate this. 

The following results will assume scalar systems, thus $A, C_m, Q, R_m$, and $P$ are all scalar. 
\begin{lemma}
\label{composition_lemma}
Let $\mathcal{F}(.)$ be a function formed by composition (in any order) of any of the functions $f(.), g_1(.),\dots,g_M(.), \textnormal{id}(.)$
where 
$$f(P) \triangleq A^2 P + Q, \quad g_m(P) \triangleq A^2 P + Q - \frac{A^2 C_m^2 P^2}{C_m^2 P + R_m},$$ and $\textnormal{id}(.)$ is the identity function. Then:
\\(i) $\mathcal{F}(.)$ is either of the affine form 
\begin{equation}
\label{F_form1}
\mathcal{F}(P) = a P + b, \textrm{ for some } a, b \geq 0
\end{equation}
or the linear fractional form
\begin{equation}
\label{F_form2}
\mathcal{F}(P) = \frac{a P + b}{c P + d}, \textrm{ for some } a, b, c, d \geq 0 \textrm{ with } ad-bc \geq 0.
\end{equation}
\\(ii)  $ \mathcal{F}(f(P)) - \mathcal{F}(g_m(P))$ is an increasing function of $P$, for $m=1,\dots,M$. 
\end{lemma}

\begin{proof}
(i) We prove this by induction. Firstly, $\textrm{id}(P) = P$ has the form (\ref{F_form1}), $f(P) = A^2 P + Q$ has the form (\ref{F_form1}), and 
$$ g_m(P)\! =\! A^2 P + Q - \frac{A^2 C_m^2 P^2}{C_m^2 P \!+\! R_m} \! =\! \frac{(A^2 R_m \!+\! C_m^2 Q) P \!+\! R_m Q}{C_m^2 P \!+\! R_m}$$
has the form (\ref{F_form2}) since $(A^2 R_m + C_m^2 Q) R_m - R_m Q C_m^2 = A^2 R_m ^2 \geq 0$. 

Now assume that $\mathcal{F}(.)$, which is a composition of the functions  $f(.), g_1(.),\dots,g_M(.), \textrm{id}(.)$, has the form of either (\ref{F_form1}) or (\ref{F_form2}). Then we will show that $f(\mathcal{F}(P))$ and $g_l(\mathcal{F}(P)), l=1,\dots,M$ also has the form of either (\ref{F_form1}) or (\ref{F_form2}). For notational convenience, let us write $$f(P) = \bar{a} P + \bar{b}$$ for some $\bar{a}, \bar{b} \geq 0$, and $$g_l(P) = \frac{\bar{a} P + \bar{b}}{\bar{c} P + \bar{d}}$$ for  some $ \bar{a}, \bar{b}, \bar{c}, \bar{d} \geq 0$ with $\bar{a} \bar{d} - \bar{b} \bar{c} \geq 0$, which can be achieved as shown at the beginning of the proof. 

If $\mathcal{F}(.)$ has the form (\ref{F_form1}), then 
$$ f(\mathcal{F}(P)) = \bar{a} (a P + b) + \bar{b}$$
is of the form (\ref{F_form1}), and 
$$ g_l(\mathcal{F}(P)) = \frac{\bar{a} (a P + b) + \bar{b}}{\bar{c}(a P + b) + \bar{d}} = \frac{\bar{a} a P + \bar{a} b + \bar{b}}{\bar{c} a P + \bar{c} b + \bar{d}}$$ 
has the form (\ref{F_form2}), since $\bar{a} a (\bar{c} b + \bar{d}) - (\bar{a} b + \bar{b}) \bar{c} a = a (\bar{a} \bar{d} - \bar{b} \bar{c}) \geq 0$. 

If $\mathcal{F}(.)$ has the form (\ref{F_form2}), then 
$$  f(\mathcal{F}(P)) = \frac{\bar{a}(a P + b)}{c P + d} + \bar{b} = \frac{(\bar{a} a + \bar{b} c)P + \bar{a} b + \bar{b}{d}}{c P + d}$$ 
has the form (\ref{F_form2}), since $(\bar{a} a + \bar{b} c )d - (\bar{a} b + \bar{b} d) c = \bar{a}(a d - b c) \geq 0$. Finally, 
$$ g_l(\mathcal{F}(P)) = \frac{\bar{a} \left(\frac{a P + b}{c P + d} \right)+ \bar{b}}{\bar{c} \left(\frac{a P + b}{c P + d} \right)+ \bar{d}} = \frac{(\bar{a} a + \bar{b} c) P + \bar{a} b + \bar{b}{d}}{(\bar{c} a + \bar{d} c) P + \bar{c} b + \bar{d} d}$$ has the form (\ref{F_form2}), since $(\bar{a} a + \bar{b} c)(\bar{c} b + \bar{d} d) -  (\bar{a} b + \bar{b} d) (\bar{c} a + \bar{d} c) = (ad-bc) (\bar{a} \bar{d} - \bar{b} \bar{c}) \geq 0$. 
\\(ii) By part (i), we know that $\mathcal{F}(.)$ is either of the form 
(\ref{F_form1}) or (\ref{F_form2}). 
If  $\mathcal{F}(.)$ has the form (\ref{F_form1}), then 
$$\mathcal{F}(f(P)) - \mathcal{F}(g_m(P)) = a (f(P) - g_m(P))$$
will be an increasing function of $P$, since 
$$f(P) - g_m(P) = \frac{A^2 C_m^2 P^2}{C_m^2 P + R_m}$$
can be easily checked to be an increasing function of $P$. 

If $\mathcal{F}(.)$ has the form (\ref{F_form2}), then it can be verified after some algebra that 
\begin{equation*}
\begin{split}
&\frac{d}{dP} \left(\mathcal{F}(f(P))\! -\! \mathcal{F}(g_m(P)) \right) \! =\! \frac{d}{dP} \left( \frac{a f(P) \!+\! b}{c f(P)\!+\!d} \!-\! \frac{a g_m(P)\!+\!b}{c g_m(P) \!+\! d}\right)
\\&\!\!=\! \frac{(ad\!-\!bc)A^2 C_m^2 \! P (d\!+\!cQ) \!\left(C_m^2\! P (d\!+\!cQ) \!+\! 2 (d \!+\! c(A^2 \!P\!+\!Q))R_m  \right)}{(d+c(A^2 P+Q))^2 \left(C_m^2 P (d+cQ) +  (d + c(A^2 P + Q))R_m \right)^2} \\ 
& \!\! \geq 0
\end{split}
\end{equation*}
since $ad-bc \geq 0$. Hence $\mathcal{F}(f(P)) - \mathcal{F}(g_m(P))$ is an increasing function of $P$. 
\end{proof}

\begin{theorem}
\label{multi_sensor_structural_properties_thm_tx_meas}
The functions
\begin{equation*}
\begin{split}
&\phi_{m,k}(P) \triangleq \beta f(P) \!+\!  J_{k+1} (f(P))  \!-\! \beta[ \lambda_m g_m(P) \!+\! (1\!-\!\lambda_m) f(P)] \\ & \quad -\! (1\!-\!\beta) E_m  \! - \!\lambda_m J_{k+1}(g_m(P))\! -\! (1\!-\!\lambda_m) J_{k+1} ( f(P))
\end{split}
\end{equation*}
for $m=1,\dots,M, k=1,\dots,K$, are increasing functions of $P$. 
\end{theorem}

\begin{proof}
The functions are equivalent to 
\begin{equation}
\label{tx_meas_thm_equivalent_fn}
\begin{split}
 \phi_{m,k}(P)& =\beta \lambda_m [f(P) - g_m(P)] - (1-\beta) E_m \\ & \quad \quad+ \lambda_m [J_{k+1} (f(P)) - J_{k+1} (g_m(P))].
 \end{split}
 \end{equation}
As stated in the proof of Lemma \ref{composition_lemma}(ii), $f(P) - g_m(P)$ can be easily verified to be an increasing function of $P$. Thus Theorem \ref{multi_sensor_structural_properties_thm_tx_meas} will be proved if we can show that $J_k(f(P)) - J_k(g_m(P))$ is an increasing function of $P$ for all $k$ and $m$. 

In fact, we will prove the stronger statement (see Remark \ref{stronger_statement_remark}) that $J_k(\mathcal{F}(f(P))) - J_k(\mathcal{F}(g_m(P)))$ is an increasing function of $P$ for all $k$ and $m$, where $\mathcal{F}(.)$ is a function formed by composition of any of the functions $f(.), g_1(.),\dots,g_M(.), \textrm{id}(.)$.
The proof is by induction.  The case of $J_{K+1}(\mathcal{F}(f(.))) - J_{K+1}(\mathcal{F}(g_m(.))) = 0$ is clear. Now assume that for $P' \geq P$, 
\begin{equation*}
\begin{split}
&J_{k'}(\mathcal{F}(f(P'))) - J_{k'}(\mathcal{F}(g_m(P'))) \\ &\quad \quad -  J_{k'}(\mathcal{F}(f(P))) + J_{k'}(\mathcal{F}(g_m(P))) \geq 0
\end{split}
\end{equation*}
holds for $k'=K+1,K,\dots,k+1$. We have
\begin{align}
\label{Jk_induction_step_tx_meas}
& J_k(\mathcal{F}(f(P'))) - J_k(\mathcal{F}(g_m(P'))) \nonumber \\ & \quad\quad  -  J_k(\mathcal{F}(f(P))) + J_k(\mathcal{F}(g_m(P))) \nonumber \\
 & \geq \min_{(\nu_{1},\dots,\nu_{M}) } \Bigg\{\beta \Big[ \sum_{l=1}^M \nu_{l} \prod_{n \neq l} (1-\nu_{n})  \lambda_l  g_l(\mathcal{F}(f(P'))) \nonumber \\ & \quad\quad   + \Big(1-\sum_{l=1}^M \nu_{l} \prod_{n \neq l} (1-\nu_{n}) \lambda_l \Big)  f(\mathcal{F}(f(P'))) \Big] \nonumber \\
 & \quad \quad  + \sum_{l=1}^M \nu_{l} \prod_{n \neq l} (1-\nu_{n})  \lambda_l J_{k+1}( g_l(\mathcal{F}(f(P')))) \nonumber \\ & \quad\quad + \Big(1-\sum_{l=1}^M \nu_{l} \prod_{n \neq l} (1-\nu_{n}) \lambda_l \Big) J_{k+1}( f(\mathcal{F}(f(P')))) \nonumber \\
& \quad  - \beta \Big[ \sum_{l=1}^M \nu_{l} \prod_{n \neq l} (1-\nu_{n})  \lambda_l  g_l(\mathcal{F}(g_m(P'))) \nonumber \\ & \quad\quad  + \Big(1-\sum_{l=1}^M \nu_{l} \prod_{n \neq l} (1-\nu_{n}) \lambda_l \Big)  f(\mathcal{F}(g_m(P'))) \Big] \nonumber \\
& \quad \quad  - \sum_{l=1}^M \nu_{l} \prod_{n \neq l} (1-\nu_{n})  \lambda_l J_{k+1}( g_l(\mathcal{F}(g_m(P')))) \nonumber \\ & \quad\quad  - \Big(1-\sum_{l=1}^M \nu_{l} \prod_{n \neq l} (1-\nu_{n}) \lambda_l \Big) J_{k+1}( f(\mathcal{F}(g_m(P')))) \nonumber \\
& \quad  - \beta \Big[ \sum_{l=1}^M \nu_{l} \prod_{n \neq l} (1-\nu_{n})  \lambda_l  g_l(\mathcal{F}(f(P)))  \nonumber \\ & \quad\quad  + \Big(1-\sum_{l=1}^M \nu_{l} \prod_{n \neq l} (1-\nu_{n}) \lambda_l \Big)  f(\mathcal{F}(f(P))) \Big] \nonumber \\
& \quad  \quad - \sum_{l=1}^M \nu_{l} \prod_{n \neq l} (1-\nu_{n})  \lambda_l J_{k+1}( g_l(\mathcal{F}(f(P)))) \nonumber \\ & \quad\quad - \Big(1-\sum_{l=1}^M \nu_{l} \prod_{n \neq l} (1-\nu_{n}) \lambda_l \Big) J_{k+1}( f(\mathcal{F}(f(P)))) \nonumber \\
& \quad + \beta \Big[ \sum_{l=1}^M \nu_{l} \prod_{n \neq l} (1-\nu_{n})  \lambda_l  g_l(\mathcal{F}(g_m(P))) \nonumber \\ & \quad\quad  + \Big(1-\sum_{l=1}^M \nu_{l} \prod_{n \neq l} (1-\nu_{n}) \lambda_l \Big)  f(\mathcal{F}(g_m(P)))  \Big] \nonumber \\
& \quad  \quad + \sum_{l=1}^M \nu_{l} \prod_{n \neq l} (1-\nu_{n})  \lambda_l J_{k+1}( g_l(\mathcal{F}(g_m(P)))) \nonumber \\ & \quad\quad  + \Big(1-\sum_{l=1}^M \nu_{l} \prod_{n \neq l} (1-\nu_{n}) \lambda_l \Big) J_{k+1}( f(\mathcal{F}(g_m(P))))  \Bigg\}. 
\end{align}
In the minimization of (\ref{Jk_induction_step_tx_meas}) above, if the optimal $(\nu_{1}^*,\dots,\nu_{M}^*) = \mathbf{e}_0$ (recall the notation of (\ref{V_defn})), then  
\begin{align*}
& J_k(\mathcal{F}(f(P'))) - J_k(\mathcal{F}(g_m(P'))) \\ & \quad \quad -  J_k(\mathcal{F}(f(P))) + J_k(\mathcal{F}(g_m(P))) \\
 & \geq \beta \big[ f( \mathcal{F}(f(P'))) - f(\mathcal{F}(g_m(P'))) \\ & \quad \quad -  f(\mathcal{F}(f(P))) + f(\mathcal{F}(g_m(P))) \big] \\
& \quad + J_{k+1} (f( \mathcal{F}(f(P')))) - J_{k+1}(f(\mathcal{F}(g_m(P')))) \\ & \quad \quad  -  J_{k+1}(f(\mathcal{F}(f(P)))) + J_{k+1}(f(\mathcal{F}(g_m(P)))) 
 \geq 0
\end{align*}
where the last inequality holds by  Lemma \ref{composition_lemma} (ii), the induction hypothesis, and the fact that $f \circ \mathcal{F}(.)$ is a composition of functions of the form $f(.), g_1(.),\dots,g_M(.), \textrm{id}(.)$. If instead the optimal $(\nu_{1}^*,\dots,\nu_{M}^*) = \mathbf{e}_l, l=1,\dots,M$, then by a similar argument
\begin{align}
\label{induction_step_el}
& J_k(\mathcal{F}(f(P'))) - J_k(\mathcal{F}(g_m(P'))) \nonumber \\ & \quad \quad -  J_k(\mathcal{F}(f(P))) + J_k(\mathcal{F}(g_m(P))) \nonumber \\
 & \geq \beta \lambda_l \big[ g_l( \mathcal{F}(f(P'))) - g_l(\mathcal{F}(g_m(P'))) \nonumber \\ & \quad \quad -  g_l(\mathcal{F}(f(P))) + g_l(\mathcal{F}(g_m(P))) \big]  \nonumber \\
 & \quad + \beta (1-\lambda_l) \big[ f( \mathcal{F}(f(P'))) - f(\mathcal{F}(g_m(P'))) \nonumber \\ & \quad \quad -  f(\mathcal{F}(f(P))) + f(\mathcal{F}(g_m(P))) \big] \nonumber \\
& \quad + \lambda_l \big[ J_{k+1} (g_l( \mathcal{F}(f(P')))) - J_{k+1}(g_l(\mathcal{F}(g_m(P')))) \nonumber \\ & \quad \quad  -  J_{k+1}(g_l(\mathcal{F}(f(P)))) + J_{k+1}(g_l(\mathcal{F}(g_m(P)))) \big] \nonumber \\
& \quad + (1-\lambda_l) \big[ J_{k+1} (f( \mathcal{F}(f(P')))) - J_{k+1}(f(\mathcal{F}(g_m(P')))) \nonumber \\ & \quad \quad  -  J_{k+1}(f(\mathcal{F}(f(P)))) + J_{k+1}(f(\mathcal{F}(g_m(P)))) \big]
 \geq 0
\end{align}
\begin{remark}
\label{stronger_statement_remark}
The reason for proving in Theorem \ref{multi_sensor_structural_properties_thm_tx_meas} the stronger statement that $J_k(\mathcal{F}(f(P))) - J_k(\mathcal{F}(g_m(P)))$ is an increasing function of $P$, is that if we carry out the arguments in (\ref{Jk_induction_step_tx_meas}) using just $J_{k}(f(P')) - J_{k}(g_m(P')) -  J_{k}(f(P)) + J_{k}(g_m(P)) $, then in (\ref{induction_step_el}) we end up needing to show statements such as $J_{k+1} (g_l(f(P'))) - J_{k+1}(g_l(g_m(P')))  -  J_{k+1}(g_l(f(P))) + J_{k+1}(g_l(g_m(P))) \geq 0$ and $J_{k+1} (f( f(P'))) - J_{k+1}(f(g_m(P')))  -  J_{k+1}(f(f(P))) + J_{k+1}(f(g_m(P))) \geq 0$, neither of which are covered by the weaker induction hypothesis that $J_{k'}(f(P')) - J_{k'}(g_m(P')) -  J_{k'}(f(P)) + J_{k'}(g_m(P)) \geq 0$ holds for $k'=K+1,K,\dots,k+1$.
\end{remark}
\end{proof}

Theorem \ref{multi_sensor_structural_properties_thm_tx_meas} is the counterpart of Theorem \ref{multi_sensor_structural_properties_theorem}(i), for estimation schemes where measurements are transmitted.  Referring back to (\ref{J_fn_multi_sensor_tx_meas}), $ \beta f(P) \!+\!  J_{k+1} (f(P)) $ is the cost function when no sensors transmit, while $ \beta[ \lambda_m g_m(P) \!+\! (1\!-\!\lambda_m) f(P)]  +\! (1\!-\!\beta) E_m  \! + \!\lambda_m J_{k+1}(g_m(P))\! +\! (1\!-\!\lambda_m) J_{k+1} ( f(P))$ is the cost function when sensor $m$ transmits. Theorem \ref{multi_sensor_structural_properties_thm_tx_meas} thereby establishes that the cost difference  between not transmitting and sensor $m$ transmitting increases with $P$, and in particular  implies the optimality of threshold policies in the single sensor, scalar case.  This provides a theoretical justification for the variance based triggering strategy proposed in \cite{TrimpeDAndrea_journal}.

For vector systems, it is well known from Kalman filtering that when measurements are transmitted, the error covariance matrices are only partially ordered. One might hope that Theorem \ref{multi_sensor_structural_properties_thm_tx_meas} will still hold for vector systems, but  
in  general this is not the case, as the following counterexample shows. 

\begin{example}
\label{counterexample1}
 Consider the case $k=K$ and $M=1$ sensor, so that we are interested in  the function (\ref{tx_meas_thm_equivalent_fn})  with $J_{K+1}(.)=0$: 
\begin{equation*}
\begin{split}
&\phi_{1,K}(P)  =\beta \lambda_1 \textnormal{tr} [f(P) - g_1(P)] - (1-\beta) E_1 \\& = \beta \lambda_1 \textnormal{tr}[ A P C_1^T (C_1 P C_1^T + R_1)^{-1} C_1 P A^T] - (1-\beta) E_1.
\end{split}
\end{equation*}
Suppose we have a system with parameters 
$$A = \left[\begin{array}{cc} 1.1 & 0.2 \\ 0.2 & 0.8 \end{array} \right], \quad C_1 = \left[\begin{array}{cc} 1 & -0.9 \end{array} \right], $$
$Q=I$, $R_1=1$. 
Let 
$$P = \bar{P}_1 = \left[\begin{array}{cc} 7.8328 & 7.3915 \\ 7.3915 & 7.7127 \end{array} \right], \quad P' = \left[\begin{array}{cc} 7.85 & 7.40 \\ 7.40 & 7.80 \end{array} \right].$$
Then one can easily verify that $P' > P$, but that 
\begin{equation*}
\begin{split}
& \textnormal{tr}[ A P' C_1^T (C_1 P' C_1^T + R_1)^{-1} C_1 P' A^T] = 1.1970  \\ & \quad < \textnormal{tr}[ A P C_1^T (C_1 P C_1^T + R_1)^{-1} C_1 P A^T] = 1.2862,
\end{split}
\end{equation*}
so the function $\phi_{1,K}(P)$
is not an increasing function of $P$. 

For vector systems with scalar measurements,  a threshold policy was considered in \cite{TrimpeDAndrea_journal}, where a sensor $m$ would transmit if $C_m P C_m^T$ exceeded a threshold. Since $P' > P$ implies $C_m P' C_m^T > C_m P C_m^T$, the  above example also shows that such a threshold policy is in general not optimal when measurements are transmitted (under our problem formulation of minimizing a convex combination of the expected error covariance and expected energy usage). 
\end{example}

Recall the property implied by Theorem \ref{multi_sensor_structural_properties_theorem}(ii), namely that if for some $P$, sensor $m$ is scheduled to transmit, while for some larger $P'$, sensor $n$ (with $n \neq m$) is scheduled to transmit, then sensor $m$ will not transmit $\forall P'' > P'$. As illustrated below, this property  does not hold when measurements are transmitted, even for scalar systems. 

\begin{example}
\label{counterexample2}
Consider a system with 2 sensors, with parameters $A=1.1$, $C_1=1, C_2=1$, $R_1=1, R_2=2$, $Q=0.1$, $\lambda_1=0.6, \lambda_2=0.7$, $E_1=0.17, E_2=0.1$, $\beta=0.5$. Again look at the case $k=K$. Then comparing the functions $\beta f(P)$, $\beta(\lambda_1 g_1(P) + (1-\lambda_1) f(P)) + (1-\beta) E_1$, and $\beta(\lambda_2 g_2(P) + (1-\lambda_2) f(P)) + (1-\beta) E_2$ (corresponding respectively to the cases when no sensor transmits, sensor 1 transmits and sensor 2 transmits), we can  verify that the optimal strategy is for no sensor to transmit when $P < 0.5485$, sensor 2 to transmit when $0.5485 \leq P < 0.8642$, sensor 1 to transmit when $0.8642 \leq P < 3.9005$, but sensor 2 will again transmit when $P \geq 3.9005$. 
\end{example}
 
\subsection{Transmitting State Estimates or  Measurements}
There are advantages and disadvantages to both scenarios of transmitting state estimates or measurements, which we will summarize in this subsection. Sending measurements is more practical when the sensor has limited computation capabilities. Furthermore, detectability at individual sensors is not required. However, as mentioned in Section \ref{model_sec}, transmitting state estimates outperforms sending of measurements. From Table I, we can see that the optimal estimator when sending estimates outperforms sending of measurements in all cases, while the suboptimal estimator also outperforms the sending of measurements in many cases. 

The optimization problems in the infinite horizon situation are also less computationally intensive in the case where state estimates are transmitted and the remote estimator (\ref{remote_estimator_eqns_multi_sensor}) is used. As mentioned in Remark \ref{computational_remark}, the set $\mathcal{S}$  has a simple form in steady state, which in practice can be easily truncated to a finite set $\mathcal{S}^N$. On the other hand, when measurements are transmitted, the set of all possible values of the error covariance  is difficult to determine in the infinite horizon case. Hence it is difficult to discretize the set of all positive semi-definite matrices (which the error covariance matrices will be a subset of) efficiently, and the  computational complexity of the associated optimization problems can be very high. 
 
\section{Numerical studies}
\label{numerical_sec}

\subsection{Single Sensor}
\label{numerical_single_sensor_sec}
We consider an example with parameters 
$$A = \left[\begin{array}{cc} 1.1 & 0.2 \\ 0.2 & 0.8  \end{array}\right],\, C = \left[\begin{array}{cc} 1 & 1 \end{array} \right],\, Q=I,\, R=1,$$
in which case
 $$\bar{P} =  \left[\begin{array}{rr} 1.3762 & -0.9014 \\ -0.9014 & 1.1867  \end{array}\right].$$
The packet reception probability is chosen to be $\lambda=0.8$, and the transmission energy cost $E=1$. 

We first consider the finite horizon problem, with $K=5$ and $\beta=0.05$, and with the local Kalman filter operating in steady state. Figs. \ref{K5k1} and \ref{K5k2} plots respectively the optimal $\nu_1^*$ and $\nu_2^*$ (i.e. $k=1$ and $k=2$) for different values of $f^n(\bar{P})$, which we recall represents the different values that the error covariance can take. In agreement with Theorem \ref{threshold_policy_theorem}, we observe a threshold behaviour in the optimal $\nu_k^*$. In this example we have $P_{0|0}^{\textrm{th}} = f^3(\bar{P})$ and $P_{1|1}^{\textrm{th}} = f^2(\bar{P})$; the thresholds  are in general different for different values of $k$. 
\begin{figure}[htb!]
\centering 
\includegraphics[scale=0.5]{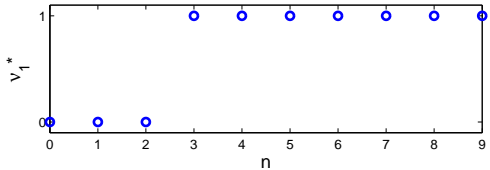} 
\caption{Finite horizon, $K=5$. $\nu_1^*$ for different values of $f^n(\bar{P})$. }
\label{K5k1}
\end{figure} 
\begin{figure}[htb!]
\centering 
\includegraphics[scale=0.5]{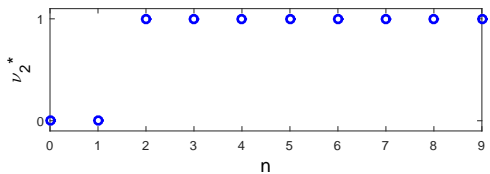} 
\caption{Finite horizon, $K=5$. $\nu_2^*$ for different values of $f^n(\bar{P})$. }
\label{K5k2}
\end{figure} 

We next consider the infinite horizon problem, with $\beta=0.05$.  Fig.  \ref{infinite_horizon} plots the optimal $\nu_k^*$ for different values of $f^n(\bar{P})$, where we again see a threshold behaviour, with $P^{\textrm{th}} = f^3(\bar{P})$.
\begin{figure}[htb!]
\centering 
\includegraphics[scale=0.5]{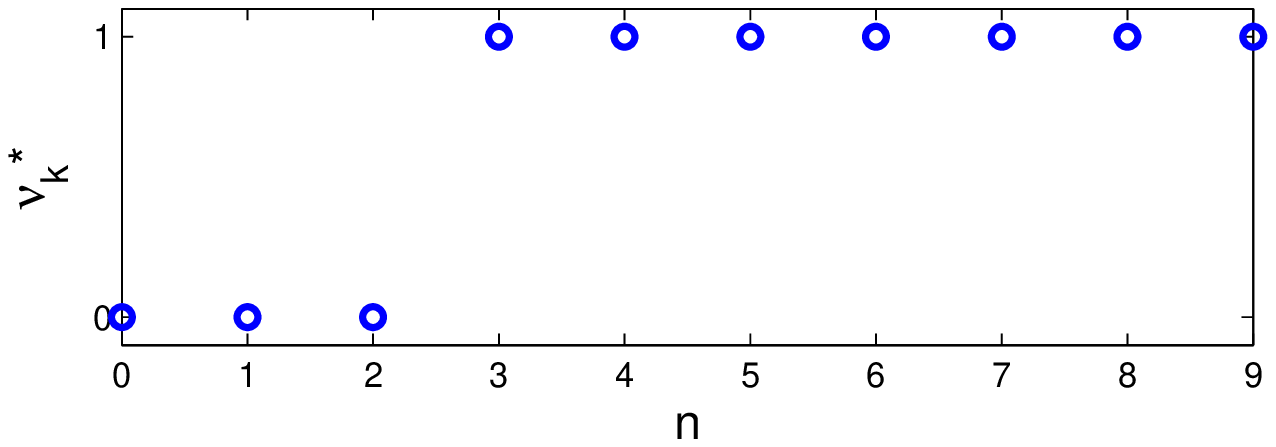} 
\caption{Infinite horizon. $\nu_k^*$ for different values of $f^n(\bar{P})$.}
\label{infinite_horizon}
\end{figure}  
In Fig.~\ref{beta_threshold_plot} we plot  the values of the thresholds for different values of $\beta$. As $\beta$ increases, the relative importance of minimizing the error covariance (vs the energy usage) is increased, thus one should transmit more often, leading to decreasing values of the thresholds. 
\begin{figure}[htb!]
\centering 
\includegraphics[scale=0.48]{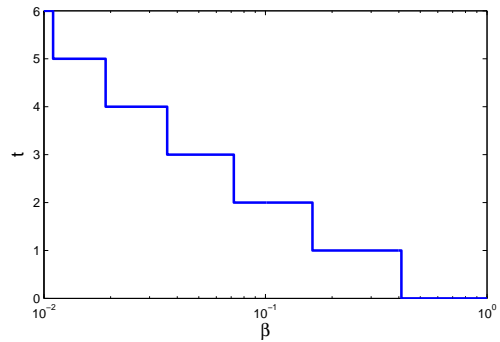} 
\caption{Infinite horizon. Threshold $P^{\textrm{th}}$ vs $\beta$, with $f^t(\bar{P})=P^{\textrm{th}}$.}
\label{beta_threshold_plot}
\end{figure}  

Finally, in Fig. \ref{energy_covariance_comparison_plot} we plot the trace of the expected error covariance vs the expected energy, obtained by solving the infinite horizon problem for different values of $\beta$, with the values computed using the expressions (\ref{expected_energy}) and (\ref{expected_error_covariance}). Note that the plot is discrete as $t \in \mathbb{N}$ in (\ref{expected_energy}) and (\ref{expected_error_covariance}), see also Fig.~\ref{beta_threshold_plot}. 
\begin{figure}[htb!]
\centering 
\includegraphics[scale=0.5]{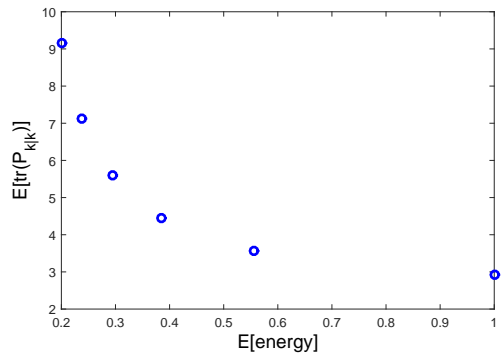} 
\caption{Infinite horizon. Expected error covariance vs expected energy.}
\label{energy_covariance_comparison_plot}
\end{figure}

\subsection{Multiple Sensors}
\label{numerical_multi_sensor_sec}
We first consider a two sensor, scalar system with parameters $A=1.1$, $C_1=1.5, C_2=1$, $Q=1$, $R_1=R_2=1$, $\lambda_1=0.8$, $\lambda_2=0.6$. We solve the infinite horizon problem with $\beta=0.2$. Fig. \ref{two_sensor_case1} plots the optimal $\nu_{1,k}^*$ and  $\nu_{2,k}^*$ for different values of $P_{k-1|k-1}$, with transmission energies $E_1=1, E_2=1$. The behaviour corresponds to scenario (i) of Corollary \ref{two_sensor_structural_lemma}.
\begin{figure}[htb!]
\centering 
\includegraphics[scale=0.5]{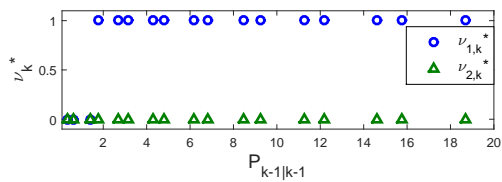} 
\caption{Infinite horizon, $\nu_{1,k}^*$ and  $\nu_{2,k}^*$ for different values of $P_{k-1|k-1}$. $E_1=1, E_2=1$.}
\label{two_sensor_case1}
\end{figure} 
\begin{figure}[htb!]
\centering 
\includegraphics[scale=0.5]{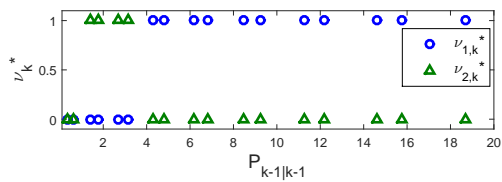} 
\caption{Infinite horizon, $\nu_{1,k}^*$ and  $\nu_{2,k}^*$ for different values of $P_{k-1|k-1}$. $E_1=1, E_2=0.4$.}
\label{two_sensor_case3}
\end{figure} 
Fig. \ref{two_sensor_case3} plots the optimal $\nu_{1,k}^*$ and  $\nu_{2,k}^*$ for different values of $P_{k-1|k-1}$, but with transmission energies $E_1=1, E_2=0.4$. With these parameters, the behaviour  corresponds to scenario (iii) of Corollary \ref{two_sensor_structural_lemma}. The remaining scenarios (ii) and (iv) of Corollary \ref{two_sensor_structural_lemma} can be illustrated  by, e.g., swapping the parameter values of sensors 1 and 2.

\subsection{Performance Comparison}
Here we will compare the performance of our approach with a scheme similar to that investigated in \cite{LiLemmonWang} (see also \cite{XiaGuptaAntsaklis}) that transmits when the difference between the state estimates at  sensor $m$ and the remote estimator exceeds a threshold $T_m$.\footnote{The scheme is not exactly the same as in \cite{LiLemmonWang} since here we also consider random packet drops.} In order to avoid collisions, which from simulation experience will greatly deteriorate performance, we allow each sensor to transmit (if it exceeds the threshold $T_m$) once every $M$ time steps in a round-robin fashion. Specifically,
\begin{equation}
\label{sensor_threshold_decision}
\begin{split}
\nu_{m,k} & = \left\{\begin{array}{lcl} 1   & , & ||\hat{x}_{m,k-1|k-1}^s - \check{x}_{k-1}|| > T_m \\ & &  \textrm{and it is sensor $m$'s turn to transmit}\\ 0 & , & \textrm{otherwise} \end{array}  \right. 
\end{split}
\end{equation} 
where 
$\check{x}_k$ is the remote estimate at time $k$.
 
  When the decisions $\nu_{m,k}$ depend on the state estimates, the optimal estimator is generally nonlinear \cite{SijsLazar,WuJiaJohanssonShi}. In the spirit of (\ref{remote_estimator_eqns_multi_sensor}), we consider a suboptimal estimator $\check{x}_k$ given by
\begin{equation}
\label{sensor_threshold_decision_estimator}
\begin{split}
\check{x}_{k} & = \left\{\begin{array}{ccc} \hat{x}_{m,k|k}^s & , & \nu_{m,k} \gamma_{m,k} = 1 \\ A \check{x}_{k-1} & , & \textrm{otherwise}.  \end{array}  \right. 
\end{split}
\end{equation} 

With this scheme the decision on whether to transmit is made by the sensor (rather than the remote estimator). The sensor has access to its local state estimate, but also requires knowledge of the remote estimate.
In the single sensor case, the sensor can reconstruct the remote estimate $\check{x}_{k-1}$ provided the values of $\gamma_{k-1}$ are fed back to the sensor before transmission at time $k$.  However, in the multiple sensor case simply feeding back $\gamma_{m,k-1}$ is not enough for the sensors to reconstruct the  remote estimate, and it appears that one requires the entire state estimate $\check{x}_{k-1}$ to be fed back to the sensors in order to implemement this scheme. Thus the scheme (\ref{sensor_threshold_decision})-(\ref{sensor_threshold_decision_estimator}) is not intended as a practical scheme for the multi-sensor case, but is only used here for performance comparison with our approach that schedules transmit decisions at the remote estimator. 

We consider the  two sensor, vector system with parameters $$A = \left[\begin{array}{cc} 1.1 & 0.2 \\ 0.2 & 0.8  \end{array}\right],\, C_1 = \left[\begin{array}{cc} 1.5 & 1.5 \end{array} \right],\,  C_2 = \left[\begin{array}{cc} 1 & 1 \end{array} \right],$$ 
$Q=I, R_1=R_2=1.$ The packet reception probabilities are  $\lambda_1=0.8$, $\lambda_2=0.6$, and the transmission energies are $E_1=1, E_2=0.4$. 
In Fig. \ref{energy_covariance_comparison_plot_two_sensor_round_robin} we plot the trace of the expected error covariance vs the expected total energy, obtained by solving the infinite horizon problem (\ref{infinite_horizon_problem_multi_sensor}) for different values of $\beta$. We compare the performance with the scheme (\ref{sensor_threshold_decision})-(\ref{sensor_threshold_decision_estimator}) for different values of the thresholds $T_1$ and $T_2$, with $T_1=T_2$. 
\begin{figure}[htb!]
\centering 
\includegraphics[scale=0.5]{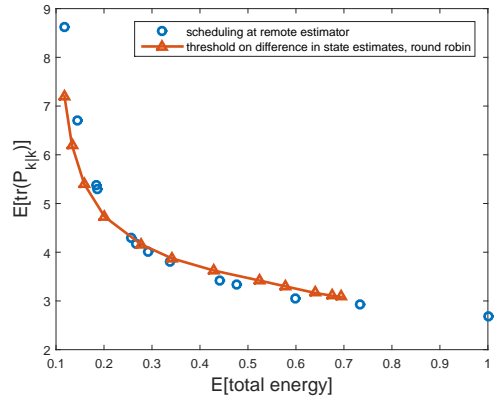} 
\caption{Infinite horizon, two sensors. Expected error covariance vs expected total energy.}
\label{energy_covariance_comparison_plot_two_sensor_round_robin}
\end{figure} 
For smaller expected energies, the scheme of (\ref{sensor_threshold_decision})-(\ref{sensor_threshold_decision_estimator}) performs better due to the utilization of additional information in the local state estimates, but as stated before requires feedback of the full remote state estimates in order to  implement. The approach proposed in Sections \ref{model_sec}-\ref{optimization_prob_sec} performs better when a smaller expected error covariance specification (with corresponding higher expected energy) is required. Furthermore, scheduling at the remote estimator doesn't require feedback of the remote estimates, but only feedback of  the decision variables $\nu_{m,k}$, which takes values of either 0 or 1 (i.e., one bit of information).

\section{Markovian packet drops}
\label{markovian_sec}
So far we have considered i.i.d. packet drops. In this section we briefly outline how our results extend to the case when state estimates are transmitted and the packet loss processes are Markovian. For notational simplicity, we restrict ourselves to the single sensor situation with the local Kalman filter operating in steady state, where the packet loss process $\{\gamma_k\}$ is a Markov chain, with transition probabilities $p \triangleq\mathbb{P}(\gamma_k=0|\gamma_{k-1}=1)$ and $q\triangleq\mathbb{P}(\gamma_k=1|\gamma_{k-1}=0)$. The probabilities $p$ and $q$ are also known as, respectively, the failure and recovery rates \cite{HuangDey}. We shall consider transmission decisions $\nu_k(P_{k-1|k-1},\gamma_{k-1})$ dependent only on $P_{k-1|k-1}$ and $\gamma_{k-1}$, in which case the remote estimator equations will still have the form 
\begin{equation*}
\begin{split}
\tilde{x}_{k|k} & = \left\{\begin{array}{ccc} \hat{x}_{k|k}^s & , & \nu_{k} \gamma_{k} = 1 \\ A \tilde{x}_{k-1|k-1} & , & \nu_{k} \gamma_{k} = 0 \end{array}  \right. \\
\tilde{P}_{k|k} & = \left\{\begin{array}{ccl} \bar{P} & , &  \nu_{k} \gamma_{k} = 1 \\ A \tilde{P}_{k-1|k-1} A^T + Q & ,  & \nu_{k} \gamma_{k} = 0. \end{array} \right. 
\end{split}
\end{equation*} 

The finite horizon problem becomes:
\begin{equation}
\label{finite_horizon_problem_markovian}
\begin{split}
 \min_{\{\nu_k\}} \sum_{k=1}^{K}  \mathbb{E}\left[  \beta \textrm{tr} P_{k|k} + (1-\beta) \nu_k E |P_{k-1|k-1}, \gamma_{k-1}, \nu_k   \right]
 \end{split}
 \end{equation}
for some $\beta \in (0,1)$, where now 
\begin{equation*}
\begin{split}
&\mathbb{E}[ \textrm{tr} P_{k|k} | P_{k-1|k-1}, \gamma_{k-1}, \nu_k] 
\\& 
= \nu_k \big(\gamma_{k-1} (1-p) 
 + (1-\gamma_{k-1})q \big) \textrm{tr} \bar{P}
 \\ & \quad + \big(1-\nu_k (\gamma_{k-1} (1-p)  + (1-\gamma_{k-1})q)\big)  \textrm{tr} f(P_{k-1|k-1}).
\end{split}
\end{equation*}
The infinite horizon problem is:
\begin{equation}
\label{infinite_horizon_problem_markovian}
 \min_{\{\nu_k\}} \limsup_{K\rightarrow \infty} \! \frac{1}{K} \! \sum_{k=1}^{K}  \mathbb{E}\!\left[ \beta \textrm{tr} P_{k|k} \! +\! (1\!-\!\beta) \nu_k E |P_{k\!-\!1|k\!-\!1},\gamma_{k\!-\!1}, \nu_k  \right].
 \end{equation}
 
The following results can be derived:
\begin{lemma}
\label{Lk_difference_lemma_markovian}
Let the functions $J_k(\cdot,\cdot): \mathcal{S} \times \{0,1\} \rightarrow \mathbb{R}$ be defined recursively for $k=1,\dots,K$ as:
\begin{equation*}
\begin{split}
&J_{K+1}(P,\gamma)  =0 \\
&J_k(P,\gamma)  = \min_{\nu \in \{0,1\}} \Big\{ \beta \big[ \nu (\gamma(1-p)+(1-\gamma)q) \textnormal{tr} \bar{P} 
\\ & \quad 
+ (1-\nu (\gamma(1-p)+(1-\gamma)q))  \textnormal{tr} f(P) \big] \\ & \quad + (1-\beta) \nu E  + \nu (\gamma(1-p)+(1-\gamma)q) J_{k+1}(\bar{P},1) 
\\ & \quad 
+ \big(1-\nu (\gamma(1-p)+(1-\gamma)q)\big) J_{k+1} (f(P),0) \Big\}.
\end{split}
\end{equation*}
and the functions $L_k^1(\cdot,\cdot) : \mathcal{S} \times \{0,1\} \rightarrow \mathbb{R}, k=1,\dots,K$ and $L_k^0(\cdot,\cdot) : \mathcal{S} \times \{0,1\} \rightarrow \mathbb{R}, k=1,\dots,K$ as: 
\begin{equation*}
\begin{split}
L_k^1(P,\nu) &\triangleq \beta \left[ \nu (1\!-\!p) \textnormal{tr} \bar{P} \! + \! (1\!-\!\nu (1\!-\!p))  \textnormal{tr} f(P) \right] 
\\ & \quad 
+ (1-\beta) \nu E  + \nu (1-p) J_{k+1}(\bar{P},1) \\ & \quad+ (1-\nu (1-p)) J_{k+1} (f(P),0) \\
L_k^0(P,\nu) &\triangleq \beta \left[ \nu q \textnormal{tr} \bar{P}  + (1-\nu q)  \textnormal{tr} (f(P)) \right] 
\\ &\quad
+ (1-\beta) \nu E  + \nu q J_{k+1}(\bar{P},1) 
\\ &\quad
+ (1-\nu q) J_{k+1} (f(P),0).
\end{split}
\end{equation*}
Then the functions $L_k^1(P,1) - L_k^1(P,0)$ and $L_k^0(P,1) - L_k^0(P,0)$ are decreasing functions of $P$. 
\end{lemma}

\begin{proof}
Similar to Lemma \ref{J_fn_increasing_lemma_multi_sensor}, we can show that $J_k(P,1)$ and $J_k(P,0)$ are both increasing functions of $P$. One can also easily verify that 
\begin{equation*}
\begin{split}
& L_k^1(P,1) \!-\! L_k^1(P,0)  \\ & = \beta (1\!-\!p) \textrm{tr} \bar{P} \!+\! (1\!-\!\beta) E  +\! (1\!-\!p) J_{k+1} (\bar{P},1) 
\\ & \quad - \beta (1-p) \textrm{tr} f(P) - (1-p) J_{k+1} (f(P),0)
\end{split}
\end{equation*}
and 
\begin{equation*}
\begin{split}
L_k^0(P,1) - L_k^0(P,0) & = \beta q \textrm{tr} \bar{P} + (1-\beta) E + q J_{k+1} (\bar{P},1) 
\\ & \quad 
- \beta q \textrm{tr} f(P) - q J_{k+1} (f(P),0)
\end{split}
\end{equation*}
which can both be shown to be decreasing functions of $P$. 
\end{proof}

Lemma \ref{Lk_difference_lemma_markovian} implies that in the finite horizon problem (\ref{finite_horizon_problem_markovian}), for each $k \in \{1,\dots,K\}$ there exist two (in general different) thresholds $P_{k-1|k-1}^{th,1}$ and $P_{k-1|k-1}^{th,0} \in \mathcal{S}$, $k=1,\dots,K$, such that when $\gamma_{k-1}=1$ then $\nu_k^*=0$ if and only if $P_{k-1|k-1} < P_{k-1|k-1}^{th,1}$; and when $\gamma_{k-1}=0$ then $\nu_k^*=0$ if and only if $P_{k-1|k-1} < P_{k-1|k-1}^{th,0}$. 

For the infinite horizon problem (\ref{infinite_horizon_problem_markovian}), arguing in a similar manner as in the proof of Lemma \ref{multi_sensor_structural_properties_theorem_inf_horizon}, the optimal policy will be such that when $\gamma_{k-1}=1$ then $\nu_k^*=0$ if and only if $P_{k-1|k-1} < P^{th,1}$; and when $\gamma_{k-1}=0$ then $\nu_k^*=0$ if and only if $P_{k-1|k-1} < P^{th,0}$, for some constant thresholds $P^{th,1}$ and $P^{th,0} \in \mathcal{S}$. 
Similar to Remark \ref{computational_remark}, knowing that the optimal policy is a threshold policy can lead to significant computational savings when solving problems (\ref{finite_horizon_problem_markovian}) and (\ref{infinite_horizon_problem_markovian}).

\section{Conclusion}
This paper has studied an event based remote estimation problem using multiple sensors, with sensor transmissions  over a shared packet dropping channel, where at most one sensor may transmit at a time. By considering an optimization problem for transmission scheduling   that minimizes a convex combination of the expected error covariance at the remote estimator and the expected energy across the sensors, we have derived structural properties on the form of the optimal solution, when either local state estimates or sensor measurements are transmitted. In particular, our results show that in the single sensor case a threshold policy is optimal.  Possible extensions of this work include the consideration of event triggered estimation with energy harvesting capabilities at the sensors \cite{NayyarBasar,Nourian_EH},  channels where multiple sensors can transmit at the same time,  and efficient ways to solve the optimal transmission scheduling problem in the case when measurements are transmitted. 
  
\begin{appendix}

\subsection{Derivation of Optimal Estimator Equations (\ref{multi_sensor_optimal_estimator})}
\label{optimal_estimator_derivation}

Note first that for the local Kalman filters at the sensors,  we have, $\forall m \in \{1,\dots,M\}$,
\begin{equation}
\label{local_KF_relation1}
\begin{split}
\hat{x}_{m,k+1|k}^s & = A \hat{x}_{m,k|k}^s \\
\hat{x}_{m,k|k}^s & = \hat{x}_{m,k|k-1}^s + K^s_{m,k} (y_{m,k} - C_m \hat{x}_{m,k|k-1}^s) 
\end{split}
\end{equation}
from which one can obtain
\begin{equation}
\label{local_KF_relation2}
\begin{split}
&x_{k+1} - \hat{x}_{m,k+1|k}^s \\& = A(I-K^s_{m,k} C_m) (x_k -\hat{x}_{m,k|k-1}^s) + w_k - A K^s_{m,k} v_{m,k}. 
\end{split}
\end{equation}

The remote estimator has the form
\begin{equation*}
\begin{split}
\hat{x}_{k+1|k} & = A \hat{x}_{k|k} \\
\hat{x}_{k|k} & = \hat{x}_{k|k-1} + \gamma_{\breve{m},k} K_{\breve{m},k} (\hat{x}_{\breve{m},k|k}^s - \hat{x}_{k|k-1})
\end{split}
\end{equation*}
when sensor $\breve{m} \in \{1,\dots,M\}$ is scheduled to transmit. We can write
\begin{equation}
\label{multi_sensor_error_relation}
\begin{split}
&x_{k+1} - \hat{x}_{k+1|k} \\ & = A x_k + w_k - A \hat{x}_{k|k-1} - \gamma_{\breve{m},k} A K_{\breve{m},k} (\hat{x}_{\breve{m},k|k}^s - \hat{x}_{k|k-1} ) \\
& = A (I -\gamma_{\breve{m},k} K_{\breve{m},k})(x_k - \hat{x}_{k|k-1}) + w_k \\ & \quad + \gamma_{\breve{m},k} A K_{\breve{m},k} (x_k - \hat{x}_{\breve{m},k|k}^s) \\
& = A (I-\gamma_{\breve{m},k} K_{\breve{m},k}) (x_k - \hat{x}_{k|k-1}) + w_k \\ & \quad \!+\! \gamma_{\breve{m},k} A K_{\breve{m},k} \left[ (I \!-\! K^s_{\breve{m},k} C_{\breve{m}}) (x_k \!-\! \hat{x}_{\breve{m},k|k-1}^s) \!-\! K^s_{\breve{m},k} v_{\breve{m},k} \right]
\end{split}
\end{equation}
where the last line comes from (\ref{local_KF_relation1}). Define $\mathcal{A}$ by (\ref{augmented_A_matrix}).
\begin{figure*}[!t]
\begin{equation}
\label{augmented_A_matrix}
\mathcal{A} = \left[\begin{array}{cccccc} 
A(I-\gamma_{\breve{m},k} K_{m,k}) & 0 & \dots &  \gamma_{\breve{m},k} A K_{\breve{m},k} (I-K^s_{\breve{m},k} C_{\breve{m}})  & \dots & 0 \\ 
0 & A(I-K^s_{1,k} C_1) & 0 & \dots &  &  0 \\
\vdots & & \ddots & &  & \vdots \\
0 & \dots & &  A (I-K^s_{\breve{m},k} C_{\breve{m}}) &  \dots & 0 \\
\vdots & & & & \ddots & \vdots \\
0 & \dots & & & & A(I-K^s_{M,k} C_M) 
 \end{array} \right]. 
 \end{equation}
\end{figure*}
 Using (\ref{multi_sensor_error_relation}) and (\ref{local_KF_relation2}), consider the augmented system
\begin{equation*}
\begin{split}
& \left[\!\! \begin{array}{c} x_{k\!+\!1} \!-\! \hat{x}_{k\!+\!1|k} \\  x_{k\!+\!1} \!-\! \hat{x}_{1,k\!+\!1|k}^s \\ \vdots \\  x_{k\!+\!1} \!-\! \hat{x}_{\breve{m},k\!+\!1|k}^s \\ \vdots \\ x_{k\!+\!1} \!-\! \hat{x}_{M,k\!+\!1|k}^s \end{array} \!\! \right] \! = \!
\mathcal{A} \!
 \left[\!\! \begin{array}{c} x_{k} \!-\! \hat{x}_{k|k\!-\!1} \\  x_{k} \!-\! \hat{x}_{1,k|k\!-\!1}^s \\ \vdots \\  x_{k} \!-\! \hat{x}_{\breve{m},k|k\!-\!1}^s \\ \vdots \\ x_{k} \!-\! \hat{x}_{M,k|k\!-\!1}^s \end{array} \!\! \right]
 \! + \!\left[\!\! \begin{array}{c} I \\ \vdots \\ \vdots \\ I \end{array}\!\! \right] \!w_k
 \\ & \quad +\! \left[ \!\!\! \begin{array}{c} \gamma_{\breve{m},k} A K_{\breve{m},k} K^s_{\breve{m},k} \\ 0 \\ \vdots \\ A K^s_{\breve{m},k} \\  \vdots \\ 0 \end{array} \!\!\! \right] \! v_{\breve{m},k} \!+\!  \left[\!\! \begin{array}{c} 0 \\ A K^s_{1,k} \\ 0 \\ \vdots \\ \vdots \\ 0 \end{array} \!\!\right] \! v_{1,k} \!+\! \dots \\ &  \quad +\! \left[\!\! \begin{array}{c} 0 \\ 0 \\ \vdots \\ \vdots \\ 0 \\ A K^s_{M,k} \end{array} \!\!\right]\! v_{M,k}.
\end{split}
\end{equation*}

Let us use the shorthand $P_k = P_{k|k-1}$. 
Then we have the recursion given in (\ref{multi_sensor_optimal_Pk_recursion}).
\begin{figure*}
\begin{equation}
\label{multi_sensor_optimal_Pk_recursion}
\begin{split}
& \left[ \begin{array}{cccc} 
P_{k+1} & P_{01,k+1} & \dots & P_{0M,k+1} \\ P_{10,k+1} & P_{11,k+1} & \dots & P_{1M,k+1} \\ \vdots &   \vdots & \ddots & \vdots \\  P_{M0,k+1} &  P_{M1,k+1} & \dots & P_{MM,k+1} \end{array} \right]  = \mathcal{A}
 \left[ \begin{array}{cccc} 
P_{k} & P_{01,k} & \dots & P_{0M,k} \\ P_{10,k} & P_{11,k} & \dots & P_{1M,k} \\ \vdots &   \vdots & \ddots & \vdots \\  P_{M0,k} &  P_{M1,k} & \dots & P_{MM,k} \end{array} \right] 
\mathcal{A}^T
 \\ & + \! \left[ \begin{array}{ccc} Q & \dots  & Q \\ \vdots  & \ddots &  \vdots  \\ Q &  \dots & Q \end{array} \right] 
\! + \! \left[ \begin{array}{cccccc} \gamma_{\breve{m},k} A K_{\breve{m},k} K^s_{\breve{m},k} R_{\breve{m}} K^{sT}_{\breve{m},k} K_{\breve{m},k}^T A^T & 0 & \dots & \gamma_{\breve{m},k} A K_{\breve{m},k} K^s_{\breve{m},k} R_{\breve{m}} K^{sT}_{\breve{m},k} A^T & \dots & 0  \\ 0 & 0 & \dots & 0 & \dots & 0 \\ \vdots & & & \vdots & & \vdots \\   \gamma_{\breve{m},k} A K^s_{\breve{m},k} R_{\breve{m}} K^{sT}_{\breve{m},k} K_{\breve{m},k}^T A^T &  0 & \dots & A  K^s_{\breve{m},k} R_{\breve{m}} K^{sT}_{\breve{m},k} A^T & \dots & 0 \\ \vdots & \vdots & & \vdots & & \vdots \\ 0 & 0 & \dots & 0 & \dots & 0 \end{array} \right] 
\\&  +  \left[ \begin{array}{cccc} 0 & 0 & \dots & 0 \\ 0 & A K^s_{1,k} R_1 K^{sT}_{1,k} A^T & \dots & 0  \\ \vdots & \vdots & \ddots & \vdots \\ 0 & 0 & \dots & 0 \end{array} \right]  + \dots + \left[ \begin{array}{cccc} 0 & 0 & \dots & 0 \\ 0 & 0 & \dots & 0 \\ \vdots & \vdots & \ddots & \vdots \\ 0 & 0 & \dots & A K^s_{M,k} R_M K^{sT}_{M,k} A^T \end{array} \right]. 
\end{split}
\end{equation}
\end{figure*}
The recursions for $P_k, P_{0m,k}, P_{mn,k}$ in the optimal estimator equations (\ref{multi_sensor_optimal_estimator}) can then be extracted from (\ref{multi_sensor_optimal_Pk_recursion}) and (\ref{augmented_A_matrix}). 
It remains to determine the optimal gains $K_{\breve{m},k}$. When $\gamma_{\breve{m},k}=0$, we have $P_{k|k} = P_k$ irrespective of $K_{\breve{m},k}$. When $\gamma_{\breve{m},k}=1$, we have 
\begin{equation*}
\begin{split}
&P_{k|k}  =  (I \!-\!  K_{\breve{m},k}) P_{k} (I \!-\!  K_{\breve{m},k})^T \!  + \! (I\!-\! K_{\breve{m},k}) P_{0\breve{m},k} \\ &\!\!\times\! (I \!-\! K^s_{\breve{m},k} C_{\breve{m}})^T \! K_{\breve{m},k}^T \!  + \!  K_{\breve{m},k} (I \!- \!K^s_{\breve{m},k} C_{\breve{m}}) P_{0\breve{m},k}^T (I \!-\! K_{\breve{m},k})^T  \\
& \quad+   K_{\breve{m},k} (I-K^s_{\breve{m},k} C_{\breve{m}}) P_{\breve{m},k|k}^s (I-K^s_{\breve{m},k} C_{\breve{m}})^T K_{\breve{m},k}^T \\ & \quad +   K_{\breve{m},k} K^s_{\breve{m},k} R_{\breve{m}} K^{sT}_{\breve{m},k} K_{\breve{m},k}^T  \\ 
& =  K_{\breve{m},k} \Big(P_k - P_{0\breve{m},k} (I - K^s_{\breve{m},k} C_{\breve{m}})^T - (I - K^s_{\breve{m},k} C_{\breve{m}}) P_{0\breve{m},k}^T \\ & \quad  + (I-K^s_{\breve{m},k} C_{\breve{m}}) P_{\breve{m},k|k}^s (I - K^s_{\breve{m},k} C_{\breve{m}})^T \\ & \quad + \! K^s_{\breve{m},k} R_{\breve{m}} K^{sT}_{\breve{m},k} \Big) K_{\breve{m},k}^T \\ & \quad + \! K_{\breve{m},k} \Big(\!-\!P_k \! + \!(I\!-\!K^s_{\breve{m},k} C_{\breve{m}}) P_{0\breve{m},k}^T  \Big) \\ & \quad +\! \Big(\!-\!  P_k \!+\!  P_{0\breve{m},k} (I\!-\!K^s_{\breve{m},k} C_{\breve{m}})^T  \Big) K_{\breve{m},k}^T \! +\!  P_k   
\end{split}
\end{equation*}
Choosing $K_{\breve{m},k}$ to minimize the expression for $P_{k|k}$, e.g. by differentiating $\textrm{tr} P_{k|k}$ with respect to $K_{\breve{m},k}$ (see \cite{Simon_book}), we  find that $ K_{\breve{m},k} = I$ if 
\begin{equation*}
\begin{split}
& P_{k|k-1} \!-\! P_{0\breve{m},k}(I\!-\!K^s_{\breve{m},k} C_{\breve{m}})^T \!-\! (I\!-\!K^s_{\breve{m},k} C_{\breve{m}}) P_{0\breve{m},k}^T \! \\ & \quad + \!(I\!-\!K^s_{\breve{m},k} C_{\breve{m}}) P_{\breve{m},k|k}^s (I\!-\!K^s_{\breve{m},k} C_{\breve{m}})^T \! +\! K^s_{\breve{m},k} R_{\breve{m}} K^{sT}_{\breve{m},k} \\&  = P_{k|k-1} - P_{0\breve{m},k} (I - K^s_{\breve{m},k} C_{\breve{m}})^T
\end{split}
\end{equation*}
and
\begin{equation*}
\begin{split}
& K_{\breve{m},k}  = \Big(P_k -  P_{0\breve{m},k}(I - K^s_{\breve{m},k} C_{\breve{m}})^T\Big) \Big(P_k - P_{0\breve{m},k} \\ & \quad \times (I - K^s_{\breve{m},k} C_{\breve{m}})^T - (I - K^s_{\breve{m},k} C_{\breve{m}}) P_{0\breve{m},k}^T \\ & \quad + (I-K^s_{\breve{m},k} C_{\breve{m}}) P_{\breve{m},k|k}^s (I - K^s_{\breve{m},k} C_{\breve{m}})^T  + K^s_{\breve{m},k} R_{\breve{m}} K^{sT}_{\breve{m}k} \Big)^{-1} 
\end{split}
\end{equation*}
otherwise.

The equations (\ref{remote_estimator_no_tx}) when no sensors are scheduled to transmit can be obtained by e.g. setting $\gamma_{\breve{m},k} = 0$ in (\ref{multi_sensor_optimal_estimator}). 

\subsection{Proof of Theorem \ref{fixed_point_theorem}}
\label{fixed_point_theorem_proof}
We first note that if $K_m=I$, then the first equation of (\ref{fixed_point_equations}) becomes
\begin{equation*}
\begin{split}
P &= (1-\lambda_m) A P A^T + Q  + \lambda_m A K^s_m R_m (K^s_m)^T A^T \\ & \quad + \lambda_m A (I-K^s_m C_m) \bar{P}_m^s (I-K^s_m C_m)^T A^T \\ & = \sqrt{(1\!-\!\lambda_m)} A P \sqrt{(1\!-\!\lambda_m)} A^T \!\!+\! Q \!+\! \lambda_m A K^s_m R_m (K^s_m)^T \!A^T \\ & \quad + \lambda_m A (I-K^s_m C_m) \bar{P}_m^s (I-K^s_m C_m)^T A^T,
\end{split}
\end{equation*}
which is a Lyapunov equation, that has a unique solution $P$ if either (i) $A$ is stable, or (ii) $A$ is unstable but with
$$\lambda_m > 1 - \frac{1}{\max_{i} |\sigma_i(A)|^2}.$$

Next, we will show that the second equation of (\ref{fixed_point_equations}) also has a solution $P_{0m} = \bar{P}_m^s$, irrespective of the value of $K_m$. 
We begin by recalling the following expressions for the error covariance and Kalman gain for the local Kalman filter at sensor $m$:
\begin{equation}
\label{local_KF_covariance_relations}
\begin{split}
\bar{P}_m^s & = A \bar{P}_m^s A^T \!+\! Q\! -\! A \bar{P}_m^s C_m^T (C_m \bar{P}_m^s C_m^T \!+\! R_m)^{-1} C_m^T \bar{P}_m^s A^T \\
K^s_m & = \bar{P}_m^s C_m^T (C_m \bar{P}_m^s C_m^T + R_m)^{-1}. 
\end{split}
\end{equation}
Since we can use (\ref{local_KF_covariance_relations}) to show that
\begin{equation*}
\begin{split}
&A \bar{P}_m^s (I \!-\!K^s_m C_m)^T A^T \!\! +\! Q =
  A \bar{P}_m^s A^T \!\!-\! A \bar{P}_m^s  C_m^T K^{sT}_{m} A^T \!\!+\! Q \\
& \!=\!  A \bar{P}_m^s A^T \!\!- \!A \bar{P}_m^s C_m^T (C_m \bar{P}_m^s C_m^T \!+\! R_m)^{-1} C_m^T \bar{P}_m^s A^T\!\! +\! Q \!=\! \bar{P}_m^s
\end{split}
\end{equation*}
and
\begin{equation*}
\begin{split}
& K^s_m C_m \bar{P}_m^s (I - K^s_m C_m)^T - K^s_m R K^{sT}_{m}  \\ &  = K^s_m C_m \bar{P}_m^s - K^s_m C_m \bar{P}_m^s  C_m^T K^{sT}_{m} - K^s_m R K^{sT}_{m}\\
& =\! K^s_m (C_m \bar{P}_m^s C_m^T \!+\! R_m) K^{sT}_{m} \!-\!  K^s_m C_m \bar{P}_m^s  C_m^T K^{sT}_{m} \\ & \quad \!-\! K^s_m R K^{sT}_{m} = 0,
\end{split}
\end{equation*}
we have
\begin{equation*}
\begin{split}
&A (I\!-\!\lambda_{m} K_{m}) P_{m}^s (I\!-\!K^s_m C_m)^T \! A^T \!\!+\! \lambda_{m} A K_{m} (I\!-\!K^s_m C_m) \\ & \quad \times \bar{P}_m^s (I \!-\! K^s_m C_m)^T \! A^T \!\!+\! Q\! +\! \lambda_{m} A K_{m} K^s_m R_m K^{sT}_{m} A^T \\
& =\! A \bar{P}_m^s (I \!-\!K^s_m C_m)^T \! A^T \!\!+\! Q \!-\! \lambda_m A K_m \big[ \bar{P}_m^s (I \!-\! K^s_m C_m)^T\!\! \\ & \quad- \!(I\!-\! K^s_m C_m) \bar{P}_m^s (I \!-\! K^s_m C_m)^T  \!-\! K^s_m R_m K^{sT}_{m} \big] A^T \\
& = \! A \bar{P}_m^s (I -K^s_m C_m)^T A^T  + Q - \lambda_m A K_m \big[K^s_m C_m \bar{P}_m^s \\ & \quad \times (I - K^s_m C_m)^T - K^s_m R K^{sT}_{m} \big] A^T  = \bar{P}_m^s.
\end{split}
\end{equation*}
Thus the equation
\begin{equation*}
\begin{split}
P_{0m} & = A (I-\lambda_{m} K_{m}) P_{0m} (I-K^s_m C_m)^T A^T \\ & \quad + \lambda_{m} A K_{m} (I-K^s_m C_m) \bar{P}_m^s (I - K^s_m C_m)^T A^T \\ &\quad + Q + \lambda_{m} A K_{m} K^s_m R_m K^{sT}_{m} A^T 
\end{split}
\end{equation*}
has $P_{0m} = \bar{P}_m^s$ as a fixed point (irrespective of the value of $K_m$). Since for the local Kalman filters $\max_{i} |\sigma_i(A(I-K^s_m C_m))| < 1$, and by assumption $\max_{i} |\sigma_i(A(I-\lambda_m K_m))| < 1$, uniqueness of the fixed point $P_{0m} = \bar{P}_m^s$ can be shown by a similar argument as in p.65 of \cite{AndersonMoore}. 

It remains to show that $K_m=I$. With $P_{0m} = \bar{P}_m^s$,  we now have from (\ref{fixed_point_equations}) that 
\begin{equation*}
\begin{split}
K_{m} & = \Big(P - \bar{P}_{m}^s (I - K^s_m C_m)^T\Big)   \Big(P - \bar{P}_m^s (I-K^s_m C_m)^T \\ & \quad - (I-K^s_m C_m) \bar{P}_m^s  + (I-K^s_m C_m) \bar{P}_m^s (I-K^s_m C_m)^T \\ & \quad + K^s_m R_m K^{sT}_{m} \Big)^{-1}.
\end{split}
\end{equation*}
Similar to above, we can show that
\begin{equation*}
\begin{split}
& - (I-K^s_m C_m) \bar{P}_m^s + (I-K^s_m C_m) \bar{P}_m^s (I-K^s_m C_m)^T \\ & \quad \quad + K^s_m R_m K^{sT}_{m} \\
&\quad= - (I-K^s_m C_m) \bar{P}_m^s C_m^T K^{sT}_{m} + K^s_m R_m K^{sT}_{m} = 0
\end{split}
\end{equation*}
and hence 
$$K_{m}  = (P - \bar{P}_{m}^s (I - K^s_m C_m)^T) (P - \bar{P}_{m}^s (I - K^s_m C_m)^T)^{-1} = I.$$

\subsection{Proof of Theorem \ref{Bellman_eqn_lemma_multi_sensor}}
\label{Bellman_eqn_lemma_multi_sensor_proof}
We will verify the conditions (CAV*1) and (CAV*2) given in Corollary 7.5.10 of \cite{Sennott_book}, which guarantee the existence of solutions to the Bellman equation for  average cost problems with countably infinite state space.  Condition (CAV*1) says that there exists a standard policy\footnote{$d$ is a \emph{standard policy} if there exists a state $z$ such that the expected first passage time $\tau_{i,z}$ from $i$ to $z$ satisfies $\tau_{i,z} < \infty, \forall i \in S$, and the expected first passage cost $c_{i,z}$ from $i$ to $z$ satisfies $c_{i,z} < \infty, \forall i \in S$.} $d$ such that the recurrent class $R_d$ of the Markov chain induced by $d$ is equal to the whole state space $S$.  Condition (CAV*2) says that given $U>0$, the set $D_U = \{i \in S| c(i,a) \leq U \textrm{ for some } a\}$ is finite, where $c(i,a)$ is the cost  at each stage when in state $i$ and using action $a$.

We first restrict ourselves to the case of a single sensor $m$. 
To verify (CAV*1), let $d$ be the policy that always transmits, i.e. $\nu_{m,k}=1, \forall k$. 
Let state $i$ of the induced Markov chain correspond to the value $f^i(\bar{P}_m), i = 0, 1, 2, \dots$, where we define $f^0(\bar{P}) \triangleq \bar{P}_m$. The state diagram of the induced Markov chain is given in Fig \ref{Markov_chain_always_transmit_sensor_m}, with state space $S = \{0,1,2,\dots\}$. 
\begin{figure}[htb!]
\centering 
\includegraphics[scale=0.4]{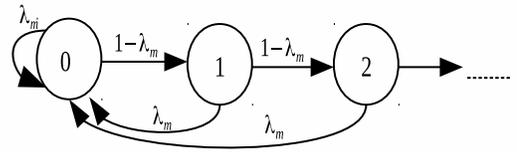} 
\caption{Markov chain for policy of always transmitting}
\label{Markov_chain_always_transmit_sensor_m}
\end{figure}

 Let $z= 0$. Then the expected first passage time from state $i$ to state $z=0$ is
\begin{equation*}
\begin{split}
\tau_{i,z} & = \lambda_m \!+\! 2 (1\!-\!\lambda_m) \lambda_m \!+\! 3 (1\!-\!\lambda_m)^2 \lambda_m \!+\! \dots = \frac{1}{\lambda_m} < \infty.
\end{split}
\end{equation*}
The expected cost of a first passage from state $i$ to state $z=0$ is
\begin{equation}
\label{first_passage_cost}
\begin{split}
&c_{i,z}  = \beta \textrm{tr} f^i(\bar{P}_m) + (1-\beta) E_m + (1-\lambda_m)  c_{(i+1),0}  \\
& = \beta \textrm{tr} f^i(\bar{P}_m) +  (1-\beta) E_m + (1-\lambda_m) \big[\beta \textrm{tr} f^{i+1}(\bar{P}_m) \\ & \,\, +  (1\!-\!\beta) E_m \big] \! + \!(1\!-\!\lambda_m)^2 \left[\beta \textrm{tr} f^{i+2}(\bar{P}_m) \!+\!  (1\!-\!\beta) E_m \right]\! +\! \dots \\
& = \beta \sum_{n=0}^\infty (1-\lambda_m)^n \textrm{tr} f^{i+n}(\bar{P}_m) + \frac{ (1-\beta) E_m}{\lambda_m}.
\end{split}
\end{equation}

For stable $A$, the infinite series above always converges. To show convergence of the infinite series for unstable $A$, note that the scenario where sensor $m$ always transmits to the remote estimator, with packet reception probability $\lambda_m$, corresponds to the situation studied in \cite{Schenato,XuHespanha}. 
By computing the stationary probabilities of the Markov chain in Fig. \ref{Markov_chain_always_transmit_sensor_m}, we can show that the expected error covariance  $\mathbb{E}[P_{k|k}]$ can be written as
$\mathbb{E}[P_{k|k}] = \sum_{n=0}^\infty (1-\lambda_m)^n \lambda_m  f^n (\bar{P}_m)$. 
From the stability results of  \cite{Schenato,XuHespanha}, we know that $\mathbb{E}[P_{k|k}]$ is bounded if and only if $\lambda_m > 1 - \frac{1}{\max_{i} |\sigma_i(A)|^2}$. Thus 
\begin{equation*}
\begin{split}
& \beta \sum_{n=0}^\infty (1-\lambda_m)^n \textrm{tr} f^{i+n} (\bar{P}_m) \\ &  = \frac{\beta}{(1-\lambda_m)^i \lambda_m} \sum_{n=0}^\infty (1-\lambda_m)^{i+n} \lambda_m \textrm{tr} f^{i+n} (\bar{P}_m) < \infty
\end{split}
\end{equation*}
when  $\lambda_m > 1 - \frac{1}{\max_{i} |\sigma_i(A)|^2}$.

Hence $d$ is a standard policy. Furthermore, one can see from Fig. \ref{Markov_chain_always_transmit_sensor_m} that the positive recurrent class $R_d$ of the induced Markov chain is equal to $S$, which verifies (CAV*1). 

Since the cost per stage $c(i,a)$ corresponds to
$\beta \textrm{tr}\tilde{P}_{k|k} + (1-\beta)  \nu_{m,k} E_m$, condition (CAV*2) can also be easily verified. This thus proves the existence of solutions to the infinite horizon problem in the case of a single sensor $m$. 

For the general case with multiple sensors, if at least one sensor $m'$ satisfies $\lambda_{m'} > 1 - \frac{1}{\max_{i} |\sigma_i(A)|^2}$, then solutions to the infinite horizon problem will exist, since restricting to this sensor $m'$ already guarantees the existence of solutions.




\end{appendix}

\bibliography{IEEEabrv,event_triggered}

\begin{thebibliography}{10}
\providecommand{\url}[1]{#1}
\csname url@samestyle\endcsname
\providecommand{\newblock}{\relax}
\providecommand{\bibinfo}[2]{#2}
\providecommand{\BIBentrySTDinterwordspacing}{\spaceskip=0pt\relax}
\providecommand{\BIBentryALTinterwordstretchfactor}{4}
\providecommand{\BIBentryALTinterwordspacing}{\spaceskip=\fontdimen2\font plus
\BIBentryALTinterwordstretchfactor\fontdimen3\font minus
  \fontdimen4\font\relax}
\providecommand{\BIBforeignlanguage}[2]{{%
\expandafter\ifx\csname l@#1\endcsname\relax
\typeout{** WARNING: IEEEtran.bst: No hyphenation pattern has been}%
\typeout{** loaded for the language `#1'. Using the pattern for}%
\typeout{** the default language instead.}%
\else
\language=\csname l@#1\endcsname
\fi
#2}}
\providecommand{\BIBdecl}{\relax}
\BIBdecl

\bibitem{LeongDeyQuevedo_ECC}
A.~S. Leong, S.~Dey, and D.~E. Quevedo, ``On the optimality of threshold
  policies in event triggered estimation with packet drops,'' in \emph{Proc.
  Europ. Contr. Conf.}, Linz, Austria, Jul. 2015, pp. 921--927.

\bibitem{XuHespanha_comm_logic}
Y.~Xu and J.~P. Hespanha, ``Optimal communication logic in networked control
  systems,'' in \emph{Proc. {IEEE} Conf. Decision and Control}, Paradise
  Islands, Bahamas, Dec. 2004, pp. 842--847.

\bibitem{ImerBasar}
O.~C. Imer and T.~Ba\c{s}ar, ``Optimal estimation with limited measurements,''
  in \emph{Proc. {IEEE} Conf. Decision and Control}, Seville, Spain, Dec. 2005,
  pp. 1029--1034.

\bibitem{CogillLallHespanha}
R.~Cogill, S.~Lall, and J.~P. Hespanha, ``A constant factor approximation
  algorithm for event-based sampling,'' in \emph{Proc. American Contr. Conf.},
  New York City, Jul. 2007, pp. 305--311.

\bibitem{LiLemmonWang}
L.~Li, M.~Lemmon, and X.~Wang, ``Event-triggered state estimation in vector
  linear processes,'' in \emph{Proc. American Contr. Conf.}, Baltimore, MD,
  Jun. 2010, pp. 2138--2143.

\bibitem{TrimpeDAndrea_IFAC}
S.~Trimpe and R.~D'Andrea, ``An experimental demonstration of a distributed and
  event-based state estimation algorithm,'' in \emph{Proc. {IFAC} World
  Congress}, Milon, Italy, Aug. 2011, pp. 8811--8818.

\bibitem{WeimerAraujoJohansson}
J.~Weimer, J.~Ara\'{u}jo, and K.~H. Johansson, ``Distributed event-triggered
  estimation in networked systems,'' in \emph{Proc. {IFAC} Conf. Analysis and
  Design of Hybrid Systems}, Eindhoven, Netherlands, Jun. 2012, pp. 178--185.

\bibitem{SijsLazar}
J.~Sijs and M.~Lazar, ``Event based state estimation with time synchronous
  updates,'' \emph{{IEEE} Trans. Autom. Control}, vol.~57, no.~10, pp.
  2650--2655, Oct. 2012.

\bibitem{TrimpeDAndrea_journal}
S.~Trimpe and R.~D'Andrea, ``Event-based state estimation with variance-based
  triggering,'' \emph{{IEEE} Trans. Autom. Control}, vol.~59, no.~12, pp.
  3266--3281, Dec. 2014.

\bibitem{XiaGuptaAntsaklis}
M.~Xia, V.~Gupta, and P.~J. Antsaklis, ``Networked state estimation over a
  shared communication medium,'' in \emph{Proc. American Contr. Conf.},
  Washington, DC, Jun. 2013, pp. 4134--4319.

\bibitem{WuJiaJohanssonShi}
J.~Wu, Q.-S. Jia, K.~H. Johansson, and L.~Shi, ``Event-based sensor data
  scheduling: Trade-off between communication rate and estimation quality,''
  \emph{{IEEE} Trans. Autom. Control}, vol.~58, no.~4, pp. 1041--1046, Apr.
  2013.

\bibitem{Han_event_trigger}
D.~Han, Y.~Mo, J.~Wu, B.~Sinopoli, and L.~Shi, ``Stochastic event-triggered
  sensor scheduling for remote state estimation,'' in \emph{Proc. {IEEE} Conf.
  Decision and Control}, Florence, Italy, Dec. 2013, pp. 6079--6084.

\bibitem{Trimpe_CDC}
S.~Trimpe, ``Stability analysis of distributed event-based state estimation,''
  in \emph{Proc. {IEEE} Conf. Decision and Control}, Los Angeles, CA, Dec.
  2014, pp. 2013--2019.

\bibitem{AstromBernhardsson}
K.~J. \r{A}str\"{o}m and B.~M. Bernhardsson, ``Comparison of {Riemann} and
  {Lebesgue} sampling for first order stochastic systems,'' in \emph{Proc.
  {IEEE} Conf. Decision and Control}, Las Vegas, NV, Dec. 2002, pp. 2011--2016.

\bibitem{Tabuada}
P.~Tabuada, ``Event-triggered real-time scheduling of stabilizing control
  tasks,'' \emph{{IEEE} Trans. Autom. Control}, vol.~52, no.~9, pp. 1680--1685,
  Sep. 2007.

\bibitem{RabiJohansson}
M.~Rabi and K.~H. Johansson, ``Scheduling packets for event-triggered
  control,'' in \emph{Proc. Europ. Contr. Conf.}, Budapest, Hungary, Aug. 2009,
  pp. 3779--3784.

\bibitem{RameshSandbergJohansson}
C.~Ramesh, H.~Sandberg, and K.~H. Johansson, ``Design of state-based schedulers
  for a network of control loops,'' \emph{{IEEE} Trans. Autom. Control},
  vol.~58, no.~8, pp. 1962--1975, Aug. 2012.

\bibitem{Quevedo_event_trigger}
D.~E. Quevedo, V.~Gupta, W.-J. Ma, and S.~Y\"{u}ksel, ``Stochastic stability of
  event-triggered anytime control,'' \emph{{IEEE} Trans. Autom. Control},
  vol.~59, no.~12, pp. 3373--3379, Dec. 2014.

\bibitem{Sinopoli}
B.~Sinopoli, L.~Schenato, M.~Franceschetti, K.~Poolla, M.~I. Jordan, and S.~S.
  Sastry, ``Kalman filtering with intermittent observations,'' \emph{{IEEE}
  Trans. Autom. Control}, vol.~49, no.~9, pp. 1453--1464, September 2004.

\bibitem{GuptaChungHassibiMurray}
V.~Gupta, T.~H. Chung, B.~Hassibi, and R.~M. Murray, ``On a stochastic sensor
  selection algorithm with applications in sensor scheduling and sensor
  coverage,'' \emph{Automatica}, vol.~42, no.~2, pp. 251--260, 2006.

\bibitem{ShiChengChen}
L.~Shi, P.~Cheng, and J.~Chen, ``Optimal periodic sensor scheduling with
  limited resources,'' \emph{{IEEE} Trans. Autom. Control}, vol.~56, no.~9, pp.
  2190--2195, Sep. 2011.

\bibitem{MoGarone}
Y.~Mo, E.~Garone, A.~Casavola, and B.~Sinopoli, ``Stochastic sensor scheduling
  for energy constrained estimation in multi-hop wireless sensor networks,''
  \emph{{IEEE} Trans. Autom. Control}, vol.~56, no.~10, pp. 2489--2495, Oct.
  2011.

\bibitem{ShiZhang}
L.~Shi and H.~Zhang, ``Scheduling two {Gauss-Markov} systems: An optimal
  solution for remote state estimation under bandwidth constraint,''
  \emph{{IEEE} Trans. Signal Process.}, vol.~60, no.~4, pp. 2038--2042, Apr.
  2012.

\bibitem{Huber}
M.~F. Huber, ``Optimal pruning for multi-step sensor scheduling,'' \emph{{IEEE}
  Trans. Autom. Control}, vol.~57, no.~5, pp. 1338--1343, May 2012.

\bibitem{SandbergRabiSkoglundJohansson}
H.~Sandberg, M.~Rabi, M.~Skoglund, and K.~H. Johansson, ``Estimation over
  heterogeneous sensor networks,'' in \emph{Proc. {IEEE} Conf. Decision and
  Control}, Cancun, Mexico, Dec. 2008, pp. 4898--4903.

\bibitem{MoGaroneSinopoli}
Y.~Mo, E.~Garone, and B.~Sinopoli, ``On infinite-horizon sensor scheduling,''
  \emph{Systems and Control Letters}, vol.~67, pp. 65--70, May 2014.

\bibitem{ZhaoZhangHu}
L.~Zhao, W.~Zhang, J.~Hu, A.~Abate, and C.~J. Tomlin, ``On the optimal
  solutions of the infinite-horizon linear sensor scheduling problem,''
  \emph{{IEEE} Trans. Autom. Control}, vol.~59, no.~10, pp. 2825--2830, Oct.
  2014.

\bibitem{MoSinopoliShiGarone}
Y.~Mo, B.~Sinopoli, L.~Shi, and E.~Garone, ``Infinite-horizon sensor scheduling
  for estimation over lossy networks,'' in \emph{Proc. {IEEE} Conf. Decision
  and Control}, Maui, HI, Dec. 2012, pp. 3317--3322.

\bibitem{LipsaMartins}
G.~M. Lipsa and N.~C. Martins, ``Remote state estimation with communication
  costs for first-order {LTI} systems,'' \emph{{IEEE} Trans. Autom. Control},
  vol.~56, no.~9, pp. 2013--2025, Sep. 2011.

\bibitem{NayyarBasar}
A.~Nayyar, T.~Ba\c{s}ar, D.~Teneketzis, and V.~V. Veeravalli, ``Optimal
  strategies for communication and remote estimation with an energy harvesting
  sensor,'' \emph{{IEEE} Trans. Autom. Control}, vol.~58, no.~9, pp.
  2246--2260, Sep. 2013.

\bibitem{XuHespanha}
Y.~Xu and J.~P. Hespanha, ``Estimation under uncontrolled and controlled
  communications in networked control systems,'' in \emph{Proc. {IEEE} Conf.
  Decision and Control}, Seville, Spain, December 2005, pp. 842--847.

\bibitem{Schenato}
L.~Schenato, ``Optimal estimation in networked control systems subject to
  random delay and packet drop,'' \emph{{IEEE} Trans. Autom. Control}, vol.~53,
  no.~5, pp. 1311--1317, Jun. 2008.

\bibitem{GuptaHassibiMurray}
V.~Gupta, B.~Hassibi, and R.~M. Murray, ``Optimal {LQG} control across
  packet-dropping links,'' \emph{Systems and Control Letters}, vol.~56, pp.
  439--446, 2007.

\bibitem{Molisch}
A.~F. Molisch, \emph{Wireless Communications}, 2nd~ed.\hskip 1em plus 0.5em
  minus 0.4em\relax John Wiley \& Sons, 2011.

\bibitem{Bertsekas_DP1}
D.~P. Bertsekas, \emph{Dynamic Programming and Optimal Control, Volume {I}},
  2nd~ed.\hskip 1em plus 0.5em minus 0.4em\relax Massachusetts: Athena
  Scientific, 2000.

\bibitem{Sennott_book}
L.~I. Sennott, \emph{Stochastic Dynamic Programming and the Control of Queueing
  Systems}.\hskip 1em plus 0.5em minus 0.4em\relax New York:
  Wiley-Interscience, 1999.

\bibitem{NgoKrishnamurthy}
M.~H. Ngo and V.~Krishnamurthy, ``Optimality of threshold policies for
  transmission scheduling in correlated fading channels,'' \emph{{IEEE} Trans.
  Commun.}, vol.~57, no.~8, pp. 2474--2483, Aug. 2009.

\bibitem{ShiEpsteinMurray}
L.~Shi, M.~Epstein, and R.~M. Murray, ``Kalman filtering over a packet-dropping
  network: A probabilistic perspective,'' \emph{{IEEE} Trans. Autom. Control},
  vol.~55, no.~3, pp. 594--604, Mar. 2010.

\bibitem{HuangDey}
M.~Huang and S.~Dey, ``Stability of {Kalman} filtering with {Markovian} packet
  losses,'' \emph{Automatica}, vol.~43, pp. 598--607, 2007.

\bibitem{Nourian_EH}
M.~Nourian, A.~S. Leong, and S.~Dey, ``Optimal energy allocation for {Kalman}
  filtering over packet dropping links with imperfect acknowledgments and
  energy harvesting constraints,'' \emph{{IEEE} Trans. Autom. Control},
  vol.~59, no.~8, pp. 2128--2143, Aug. 2014.

\bibitem{Simon_book}
D.~Simon, \emph{Optimal State Estimation: Kalman, $H_{\infty}$, and Nonlinear
  Approaches}.\hskip 1em plus 0.5em minus 0.4em\relax New Jersey:
  Wiley-Interscience, 2006.

\bibitem{AndersonMoore}
B.~D.~O. Anderson and J.~B. Moore, \emph{Optimal Filtering}.\hskip 1em plus
  0.5em minus 0.4em\relax New Jersey: Prentice Hall, 1979.

\end{thebibliography}
\bibliographystyle{IEEEtran}

\end{document}